\documentclass[a4paper]{article}

\makeatother

\usepackage{lmodern}

\usepackage{amsmath}
\usepackage{amsthm}
\usepackage{amssymb}
\usepackage{xspace}
\usepackage{esvect}

\usepackage{algorithm}
\usepackage[noend]{algorithmic}

\usepackage{float}
\usepackage{graphicx}
\usepackage{tcolorbox}
\usepackage[T1]{fontenc}

\usepackage{color,soul}
\usepackage{xcolor}

\usepackage{hyperref}

\newcommand{\reals}{\mathbb{R}}

\newcommand{\eps}{\ensuremath{\epsilon}}

\newcommand{\tcal}{\ensuremath{\mathcal{T}}}

\newcommand{\BU}{\ensuremath{\mathtt{BU}}\xspace}
\newcommand{\linBU}{\ensuremath{\mathtt{LinearBU}}\xspace}
\newcommand{\FIXP}{\ensuremath{\mathtt{FIXP}}\xspace}
\newcommand{\TFNP}{\ensuremath{\mathtt{TFNP}}\xspace}
\newcommand{\linFIXP}{\ensuremath{\mathtt{LinearFIXP}}\xspace}
\newcommand{\DlinFIXP}{\ensuremath{\mathtt{2DLinearFIXP}}\xspace}
\newcommand{\PPAD}{\ensuremath{\mathtt{PPAD}}\xspace}
\newcommand{\PPA}{\ensuremath{\mathtt{PPA}}\xspace}
\newcommand{\etr}{\ensuremath{\mathtt{ETR}}\xspace}
\newcommand{\fetr}{\ensuremath{\mathtt{FETR}}\xspace}
\newcommand{\tfetr}{\ensuremath{\mathtt{TFETR}}\xspace}
\newcommand{\betr}{\ensuremath{\mathtt{ETR}_{[0,1]}}\xspace}
\newcommand{\NP}{\ensuremath{\mathtt{NP}}\xspace}
\newcommand{\FNP}{\ensuremath{\mathtt{FNP}}\xspace}
\newcommand{\PSPACE}{\ensuremath{\mathtt{PSPACE}}\xspace}
\newcommand{\p}{\ensuremath{\mathtt{P}}\xspace}

\newcommand{\borsuk}{\ensuremath{\textsc{Borsuk-Ulam}}\xspace}
\newcommand{\epsborsuk}{\ensuremath{\eps\textsc{-Borsuk-Ulam}}\xspace}
\newcommand{\linborsuk}{\ensuremath{\textsc{Linear-Borsuk-Ulam}}\xspace}
\newcommand{\CH}{\ensuremath{\textsc{Consensus Halving}}\xspace}
\newcommand{\linCH}{\ensuremath{\textsc{Linear Consensus Halving}}\xspace}

\newcommand{\feas}{\ensuremath{\textsc{Feasible}}\xspace}
\newcommand{\bfeas}{\ensuremath{\textsc{Feasible}_{[0,1]}}\xspace}

\newcommand{\bconj}{\ensuremath{\textsc{Conjuction}_{[0,1]}}\xspace}
\newcommand{\tucker}{\ensuremath{\textsc{Tucker}}\xspace}
\newcommand{\gcircuit}{\ensuremath{\textsc{GCircuit}}\xspace}

\newcount\Comments  % 0 suppresses notes to selves in text
\definecolor{darkgreen}{rgb}{0,0.6,0}
\newcommand{\kibitz}[2]{\ifnum\Comments=1{\color{#1}{#2}}\fi}

\usepackage{todonotes}

\usepackage{setspace}

\newtheorem{theorem}{Theorem}
\newtheorem{lemma}[theorem]{Lemma}
\newtheorem{definition}[theorem]{Definition}

\newtheorem{corollary}[theorem]{Corollary}

\usepackage{fullpage}

\title{Computing exact solutions of consensus halving and the Borsuk-Ulam theorem\thanks{A preliminary version of this paper appeared in the Proceedings of the 46th International Colloquium on Automata, Languages and Programming (ICALP 2019) \cite{I-DFMS19}.}}

\author{Argyrios Deligkas\thanks{Royal Holloway University of London, UK. 
		Email: \texttt{argyrios.deligkas@rhul.ac.uk}} 
	\and John Fearnley\thanks{University of Liverpool, UK. 
		Email: \texttt{john.fearnley@liverpool.ac.uk}} 
	\and Themistoklis Melissourgos\thanks{Technical University of Munich, Germany. 
		Email: \texttt{themistoklis.melissourgos@tum.de}}
	\and Paul G. Spirakis\thanks{Department of Computer Science, University of Liverpool, UK.
	Email: \texttt{p.spirakis@liv.ac.uk}} \thanks{Computer Engineering and Informatics Department, University of Patras, Greece.}
	}
\date{\vspace{-1.0cm}}

\begin{document}

\maketitle

~\\  %%% Remove for SOCG

\begin{abstract}
We study the problem of finding an exact solution to the Consensus Halving
problem. While recent work has shown that the approximate version of this
problem is \PPA-complete~\cite{FG18,FG18a}, we show that the exact version is
much harder. Specifically, finding a solution with $n$ agents and $n$ cuts is \FIXP-hard, and
deciding whether there exists a solution with fewer than $n$ cuts is
\etr-complete. 
%We also give a QPTAS for the case where each agent's valuation is
%a polynomial.

Along the way, we define a new complexity class, called \BU, which captures all problems
that can be reduced to solving an instance of the Borsuk-Ulam problem exactly.
We show that $\FIXP \subseteq \BU \subseteq \tfetr$  and that $\linBU = \PPA$,
where \linBU is the subclass of \BU in which the Borsuk-Ulam instance is
specified by a linear arithmetic circuit.
\end{abstract}

\setcounter{page}{1}

%
%%%%%%%%%%%%%%%%  START TABLE FOR ARXIV %%%%%%%%%%%%%%%%%%

\newpage

~\\

%\doublespacing
\tableofcontents
%\singlespacing

\newpage

%%%%%%%%%%%%%%%%  END TABLE FOR ARXIV %%%%%%%%%%%%%%%%%%%%

\section{Introduction}
\label{sec:intro}

%\argy{Paul suggested to write a paragraph about motivating multiagent system.}

%Assume that a company wants to perform split-up into two separately run companies.
%Ideally, the result of this split-up should be ``fair''; create two companies of
%equal value in the eyes of every shareholder and every employee. In addition, 
%the split-up should not fragment a lot the assets of the company. Is it possible?
%The Borsuk-Ulam theorem dictates that it is!

% From Simons and Su
%The study of fair division problems is concerned with finding ways to divide an object among several parties according to some notion of fairness. The cake-cutting problem of Steinhaus (1948) is perhaps the best known example. Aside from the division of goods, other fair-division problems address the division of burdens (e.g., the chore-division problem (Gardner, 1978; Peterson and Su, 2002)) and the division of mixtures of goods and burdens (e.g., the rent-partitioning problem (Brams and Kilgour, 2001; Haake et al., 2002; Su, 1999): how to split rent so that housemates are satisfied by different rooms).
%Recently, ideas from combinatorial topology have provided new and constructive methods for obtaining solutions to fair-division problems. Su (1999) discusses a cake-cutting procedure of Simmons that can be extended to obtain envy-free solutions for chore division and rent-partitioning using variants of a result known as Sperner’s lemma, which is the combinatorial equivalent of the Brouwer fixed point theorem of topology.

Dividing resources among agents in a fair manner is among the 
most fundamental problems in multi-agent systems~\cite{brams1996fair}. 
Cake cutting~\cite{amanatidis2018improved,aziz2016discrete,aziz2016any,brams1995envy},
and rent division~\cite{brams2001competitive,haake2002bidding,edward1999rental}
are prominent examples of problems that lie in this category. At their core, 
each of these problems has a desired solution whose existence
is usually proved via a theorem 
from algebraic topology such as Brouwer's fixed point theorem, Sperner's lemma,
or Kakutani's fixed point theorem.

%In 2015 The Hewlett-Packard Company completed a split-up that
%resulted in the official formation of two new entities, HP Inc. and 
%Hewlett-Packard Enterprises. With the split-up, shareholders of the 
%original company were able to choose which entity they wished to 
%remain invested in. Ideally, the new-formed companies should have the 
%same value for every shareholder. In addition, the split-up should not
%fragment a lot the assets of the company. Is this always possible? The 
%Borsuk-Ulam theorem dictates that it is!

%Aris's motivating example.
%Suppose that two families wish to split a piece of land into two regions such 
%that every member of each family believes that the land is equally divided, or
%suppose that a conference organizer wants to assign the conference presentations 
%to the morning and the afternoon sessions, so that every participant thinks that 
%the two sessions are equally interesting. Is it possible to achieve these 
%objectives? If yes, how can it be done and how efficiently? What if we aim for
%“almost equal” instead of “equal”?

%The Borsuk-Ulam theorem is a theorem from algebraic topology that underlies a
%number of \emph{fair division} problems between agents. One such problem is
%\emph{consensus halving}

In this work we focus on a fair-division problem called \emph{Consensus Halving}: an object $A$ represented by $[0, 1]$ is to be divided
into two halves $A_+$ and $A_-$, so that $n$ agents agree that $A_+$ and
$A_-$ have the same value. 
Provided the agents have bounded and
continuous valuations over $A$,
this can always be achieved using at most $n$ cuts, and this fact can be proved
via the 
Borsuk-Ulam theorem from algebraic topology~\cite{SS03}. 
The 
necklace splitting and 
ham-sandwich problems are two other examples of fair-division problems for which
the existence of a solution can be proved via the Borsuk-Ulam
theorem~\cite{alon1987splitting,alon1986borsuk,papadimitriou1994complexity}.

Recent work has
further refined the complexity status of \emph{approximate} Consensus Halving,
in which we seek a division of the object so that every agent agrees
that the values of $A_+$ and $A_-$ differ by at most $\eps$.
Since the problem always has a solution, it lies in \TFNP,
which is the class of function problems in \NP that always have a
solution. 
More recent work has shown that the problem
is \PPA-complete~\cite{FG18}, even for $\eps$ that is inverse-polynomial in
$n$~\cite{FG18a}. The problem of deciding whether there exists an approximate
solution with $k$-cuts when $k < n$ is \NP-complete~\cite{FFGZ18}. These results
are particularly notable, because they identify Consensus Halving as one of the
first natural \PPA-complete problems.

While previous work has focused on approximate solutions to the problem, in
this work we study the complexity of solving the problem \emph{exactly}. For
problems in the complexity class \PPAD, which is a subclass of both \TFNP and
\PPA, prior work has found that there is a sharp contrast between exact and
approximate solutions. For example, the Brouwer fixed point theorem is the
theorem from algebraic topology that underpins \PPAD. Finding an
approximate Brouwer fixed point is
\PPAD-complete~\cite{papadimitriou1994complexity}, but finding an exact Brouwer
fixed point is complete for (and the defining problem of) a complexity class called \FIXP~\cite{EY10}.

It is believed that \FIXP is significantly harder than \PPAD. While \PPAD
$\subseteq$ \TFNP $\subseteq$ \FNP, there is significant doubt about whether \FIXP $\subseteq$ \FNP.
One reason for this is that there are Brouwer instances for which all
solutions are irrational. This is not particularly relevant when we seek an
approximate solution, but is a major difficulty when we seek an exact solution.
For example, in the PosSLP problem, a division free arithmetic circuit with operations $+, -,*$, inputs $0$ and $1$ and a designated output gate are given, and we are asked to decide whether the integer at the output of the circuit is positive. This fundamental problem is not known to lie in \NP, and can be reduced to the
problem of finding an approximation of 3-player Nash equilibrium \cite{EY10}. Due to the aforementioned paper, the later problem reduces to the problem of finding an exact Brouwer fixed point, which provides evidence
that \FIXP may be significantly harder than \FNP.

\subsection{Contribution}

In this work we study the complexity of solving the Consensus Halving problem
exactly. In our formulation of the problem, the valuation
function of the agents is presented as an arbitrary arithmetic circuit, and the
task is to cut $A$ such that all agents agree that $A_+$ and $A_-$ have exactly
the same valuation. 
We study two problems. The $(n, n)$-\CH problem asks us to find an exact
solution for $n$-agents using at most $n$-cuts, while the $(n, k)$-\CH problem
asks us to decide whether there exists an exact solution for $n$-agents using at
most $k$-cuts, where $k < n$.

Our results for $(n, n)$-\CH are intertwined with a new complexity class
that we call \BU. This class consists of all problems that can be reduced in
polynomial time to the problem of finding a solution of the Borsuk-Ulam problem.
We show that $(n, n)$-\CH lies in \BU, and is \FIXP hard. The hardness for
\FIXP implies that the exact variant of Consensus Halving is significantly
harder than the approximate variant: while the approximate problem is
\PPA-complete, the exact variant is unlikely to be in \FNP.

We show that $(n, k)$-\CH is \etr-complete. The complexity class \etr consists
of all decision problems that can be formulated in the \emph{existential theory
	of the reals}. It is known that $\NP \subseteq \etr \subseteq
\PSPACE$~\cite{C88}, and it is generally believed that \etr is distinct from the
other two classes. So, our result again shows that the exact version of the
problem seems to be much harder than the approximate version, which is
\NP-complete~\cite{FFGZ18}.

Just as \FIXP can be thought of as the exact analogue of \PPAD, we believe that
\BU is the exact analogue of \PPA, and we provide some evidence to justify this.
It has been shown that $\linFIXP = \PPAD$~\cite{EY10}, which is
the version of the class in which arithmetic circuits are restricted to produce
piecewise \emph{linear} functions (\FIXP allows circuits to compute piecewise
polynomials). We likewise define \linBU, which consists of all problems that can
be reduced to a solution of a Borsuk-Ulam problem using a piecewise linear
function, and we show that \linBU = \PPA. 

The containment $\linBU \subseteq \PPA$ can be proved using similar techniques to
the proof that $\linFIXP \subseteq \PPAD$. However, the proof that $\PPA
\subseteq \linBU$ utilises  our \BU containment result for Consensus Halving.
In particular, when the input to Consensus Halving is a
piecewise linear function, our containment result shows that the problem actually lies in
\linBU. The \PPA-hardness results for Consensus Halving show that
piecewise-linear-Consensus Halving is \PPA-hard, which completes the
containment~\cite{FG18,FG18a}.

Let us present a roadmap of this work. In Section \ref{sec: preliminaries} we give formal definitions for the notions and models that are used throughout the paper. In Section \ref{sec:class_bu} we introduce the complexity class \BU and its linear version \linBU, and show that $\linBU \subseteq \PPA$. Then, in Section \ref{sec:bu} we focus on the Consensus Halving problem and show containment results for variations of it. Following this, in Section \ref{sec:hardness_res} the most challenging set of this paper's results is presented, namely hardness of the Consensus Halving variations in already known complexity classes. Finally, the most technical parts of the paper are presented in Sections \ref{app:circuit-embed}, \ref{app:CH_is_FIXP-hard}, \ref{app:Feas_is_etr-c}, \ref{app:CH_is_etr-c}.

%Finally, we show that, for the case where each agent has a non-piecewise
%polynomial valuation of constant (resp. $O(\log n)$) degree, an approximate solution to the problem can be
%found using $O(\log n)$ (resp. $O(poly \log n)$) cuts, which then yields a QPTAS for the problem.

\subsection{Related work}
Although for a long period there were a few results about \PPA, recently there
has been a  
flourish of \PPA-completeness results. The first \PPA-completeness result was
given by~\cite{grigni2001sperner} who showed \PPA-completeness of the Sperner problem for a 
non-orientable 3-dimensional space. 
In~\cite{friedl2006locally} this result was strengthened for a non-orientable and locally 
2-dimensional space. In~\cite{ABB15}, 2-dimensional Tucker was shown to be 
\PPA-complete; this result was used in~\cite{FG18,FG18a} to prove \PPA-completeness 
for approximate Consensus Halving. 
In~\cite{DE+16} \PPA-completeness was proven for a special version
of Tucker and for problems of the form ``given a discrete fixed point in a non-orientable 
space, find another one''. Finally, in~\cite{DFK17} it was shown that octahedral Tucker
is \PPA-complete.
%
%%%%%%%%%%%%%%%%%%
%
In~\cite{M14}, a subclass of $\DlinFIXP \subseteq \FIXP$ that consists of 2-dimensional fixed-point problems was
studied, and it was proven that $\DlinFIXP=\PPAD$.
%and via this it was furthermore proven that computing an approximate fixed point of a  
%piecewise-linear function is \PPAD-hard.

A large number of problems are now known to be \etr-complete:
geometric  intersection 
problems~\cite{mat14,scha09},  graph-drawing problems~
\cite{ADD+16,Bien91,CaHo17,scha13}, matrix factorization problems~\cite{shi16,shi17}, 
the Art Gallery problem~\cite{AAM18}, and deciding the existence of constrained (symmetric)
Nash equilibria in (symmetric) normal form games with at least three players \cite{BM12,BM14,BM16,BM17,GMV+,BH19}.

\section{Preliminaries}
\label{sec: preliminaries}

\subsection{Arithmetic circuits and reductions between real-valued search problems}\label{sec:ar_cir_red}
An arithmetic circuit is a representation of a continuous function 
$f : \reals^n \to \reals^m$. The circuit is defined by a pair $(V, \tcal)$, where
$V$ is a set of nodes and $\tcal$ is a set of gates. There are $n$ nodes in $V$
that are designated to be \emph{input nodes}, and $m$ nodes in $V$ that are
designated to be \emph{output nodes}. 
When a value $x \in \reals^n$ is presented
at the input nodes, the circuit computes values for all other nodes $v \in V$,
which we will denote as $x[v]$. The values of $x[v]$ for the $m$ output nodes
determine the value of $f(x) \in \reals^m$.

Every node in $V$, other than the input nodes, is required to be the output of
exactly one gate in $\mathcal{T}$. 
Each gate $g \in \tcal$ enforces an arithmetic
constraint on its output node, based on the values of some other node in the
circuit. Cycles are not allowed in these constraints. 
We allow the operations $\{\zeta,
+, -, *\zeta, *, \max, \min\}$, which correspond to the gates shown in
Table~\ref{tbl:gates}.
Note that every gate computes a continuous function over its inputs, and thus
any function $f$ that is represented by an arithmetic circuit of this form is
also continuous. 

%Every node $v \in V$ 
%has a value $x[v]$ and serves as an input or as an output to a gate, 
%or both. Crucially, a node can serve as an input to many gates, but 
%it serves as an output at most for one gate. Every gate has a type 
%and performs an operation; every gate outputs to a node a value that 
%depends on its type and on the values of the input nodes it has.
%The next table shows the six types of gates arithmetic circuits use and
%the value $x[v_{out}]$ they assign to their output node $v_{out}$.
%The size of circuit is defined as $|\tcal|$.

\begin{table}
	\begin{center}
		\begin{tabular}{|l|l|}
			\hline
			\textbf{Gate}                                & \textbf{Constraint} \\ \hline
			\quad $G_{\zeta}(\zeta, v_{out})$ \quad  &  \quad $x[v_{out}] = \zeta$, where $\zeta \in \mathbb{Q}$ \quad            \\ \hline
			\quad $G_{+}(v_{in1}, v_{in2}, v_{out})$ \quad  &  \quad $x[v_{out}] = x[v_{in1}] + x[v_{in2}]$ \quad           \\ \hline
			\quad $G_{-}(v_{in1}, v_{in2}, v_{out})$ \quad  &   \quad $x[v_{out}] = x[v_{in1}] - x[v_{in2}]$   \quad         \\ \hline
			\quad $G_{*\zeta}(\zeta, v_{in}, v_{out})$ \quad  &  \quad $x[v_{out}] = x[v_{in1}]\cdot \zeta$, where $\zeta \in \mathbb{Q}$ \quad            \\ \hline
			\quad $G_{*}(v_{in1}, v_{in2}, v_{out})$ \quad  &   \quad $x[v_{out}] = x[v_{in1}] \cdot x[v_{in2}]$  \quad          \\ \hline
			\quad $G_{\max}(v_{in1}, v_{in2}, v_{out})$ \quad  &   \quad $x[v_{out}] = \max \{ x[v_{in1}], x[v_{in2}]\}$   \quad         \\ \hline
			\quad $G_{\min}(v_{in1}, v_{in2}, v_{out})$ \quad  &   \quad $x[v_{out}] = \min \{ x[v_{in1}], x[v_{in2}]\}$    \quad        \\ \hline
		\end{tabular}
	\end{center}
	\caption{The types of gates and their constraints.}
	\label{tbl:gates}
\end{table}

We study two types of circuits in this work. \emph{General} arithmetic circuits
are allowed to use any of the gates that we have defined above. \emph{Linear}
arithmetic circuits allow only the operations $\{\zeta, +, -, *\zeta, \max,
\min\}$, and the $*$ operation (multiplication of two variables) is disallowed.
Observe that a linear arithmetic circuit computes a continuous, piecewise linear function.

In this work we do not deal with the usual discrete search problems, but instead we study continuous problems whose solutions are exact, and involve representation of real numbers. However, the computation model on which we work is still the discrete Turing machine model and not a computation model over the reals like the BSS machine \cite{BSS89}. For this reason, when we consider reductions from a problem $P$ to a problem $Q$ with real-valued solutions, we have to be restricted to functions $f$, $g$ with certain properties, that transform instances of $P$ to instances of $Q$ and solutions of $Q$ to solutions of $P$ respectively. In particular, $f$ and $g$ should be polynomial time computable, and $g$ has to map efficiently (discrete) solutions of $Q$ back to (discrete) solutions of $P$, for the corresponding discrete versions of $P$ and $Q$.

In the proof of Theorem \ref{thm:CHinBU} we analytically present what function $g$ is allowed to do. For both the cases where the input to problem $P$ is a set of (a) general circuits (which represent functions whose roots are possibly irrational numbers) or (b) linear circuits (where roots are rational), $g$ is implemented by a polynomial size arithmetic circuit with an additional type of gate $G_{>}$. This is a {\em comparison gate} with a single input which outputs 1 if the input is positive and 0 otherwise. We highlight that the extra gate type implements a discontinuous function, but this does not matter since the comparison gate is only used for function $g$, and not for $f$. 
Note that our reductions use more powerful functions than the ``SL-reductions'' used by Etessami and Yannakakis in \cite{EY10}, but nevertheless computable in polynomial time.
%Specifically, as it becomes apparent in the proof of Theorem \ref{thm:CHinBU}, our result on \BU containment of $(n,n)$-\CH uses arithmetic circuits with $G_{>}$ gates.

\subsection{The Consensus Halving problem}\label{sec:CH_definition}

In the Consensus Halving problem there is an object $A$ that is represented by
the $[0, 1]$ line segment, and there are $n$ agents. We wish to divide $A$ into
two (not necessarily contiguous) pieces such that every agent agrees that the
two pieces have equal value. Simmons and Su~\cite{SS03} have shown that,
provided the agents have bounded and continuous valuations over $A$, then we can
find a solution to this problem using at most $n$ cuts.

In this work we consider instances of Consensus Halving where the
valuations of the agents are presented as arithmetic circuits. Each
agent has a valuation function $f_i : [0, 1] \rightarrow \reals$, but it is
technically more convenient if they give us a representation of the
\emph{integral} of
this function. So for each agent $i$, we are given an arithmetic circuit
computing $F_i : [0, 1] \rightarrow \reals$ where for all $x \in [0, 1]$ we have
$F_i(x) = \int_0^x f(y) \, dy$. Then, the value of any particular segment of
$[a, b]$ to agent $i$ can be computed as $F_i(b) - F_i(a)$.

%A solution to Consensus Halving is given by a \emph{$k$-cut} of the
%object $A$, which is defined by a vector of \emph{cut-points} $(t_1, t_2, \dots,
%t_k) \in [0,1]^k$, where $t_1 \leq \dots \leq t_k$, and a vector of \emph{signs}
%$(s_1, s_2, \ldots, s_{k+1}) \in \{-1,+1\}^{k+1}$. The cut-points 
%$t_i$ split $A$ into up to  $k+1$ pieces. Note that they may in fact
%split $A$ into fewer than $k+1$ pieces in the case where two cut-points $t_i =
%t_j$ overlap. We define $X_i$ to be the $i$th piece of $A$, meaning that $X_0 =
%[0, t_1]$, $X_i = [t_i, t_{i+1}]$ for all $i$ in the range $1 \le i < k$, and
%$X_k = [t_{k}, 1]$.
%
%The sign vector determines which half of $A$ the piece belongs to. 
%We define $A_+ := \{X_i: s_i =+1 \}$ and
%$A_-:= \{X_i: s_i =-1 \}$ to be the two halves. For each agent $i$, we denote 
%the value $A_+$ to agent $i$ as $F_i(A_+) := \sum_{[a, b] \in A_+} \left(F_i(b)
%- F_i(a) \right)$, and we define $F_i(A_-)$ analogously. The $k$-cut is a
%solution to the Consensus Halving problem if $F_i(A_+) = F_i(A_-)$ for all
%agents $i$.

A solution to Consensus Halving is given by a \emph{$k$-cut} of the
object $A$, which is defined by a vector of \emph{cut-points} $(t_1, t_2, \dots,
t_k) \in [0,1]^k$, where $t_1 \leq \dots \leq t_k$. The cut-points 
$t_i$ split $A$ into up to  $k+1$ pieces. Note that they may in fact
split $A$ into fewer than $k+1$ pieces in the case where two cut-points $t_i =
t_j$ overlap. We define $X_i$ to be the $i$-th piece of $A$, meaning that $X_i = [t_{i-1}, t_{i}]$ for all $i \in [k+1]$, where we set $t_0 := 0$ and $t_{k+1} := 1$.

In a Consensus Halving solution the object $A$ is divided into two ``super-pieces'' $A_+$ and $A_-$ formed by the $k+1$ pieces induced by the $k$-cut. Each piece is assigned a sign $``+''$ or $``-''$ and all of the pieces with positive sign consist super-piece $A_+$ while the rest consist $A_-$.
For each agent $i$, we denote 
the value $A_+$ as $F_i(A_+) := \sum_{[a, b] \in A_+} \left(F_i(b)
- F_i(a) \right)$, and we define $F_i(A_-)$ analogously. The $k$-cut is a
solution to the Consensus Halving problem if $F_i(A_+) = F_i(A_-)$ for all
agents $i$. Without loss of generality we can consider only solutions of Consensus Halving where the signs of the pieces are alternating. That is because in any solution that has two consecutive pieces of same sign, the cut that separates them can be removed and transferred at the right end of $A$, taking value 1, and the two pieces can be merged into a single one. Throughout the paper we implicitly consider this definition of a solution. Notice that the solutions come in symmetric pairs, where the cuts are at the exact same points, and the signs are opposite.

We define two computational problems. 
Simmons and Su~\cite{SS03} have proved that there always exists a solution
using at most $n$-cuts, and our first problem is to find that solution.

%An instance of Consensus Halving consists of $n$ agents, an object $A$ which is denoted 
%by the line segment $[0,1]$, and a valuation function $f_i$ for every agent over the object.
%Every agent has zero value for any point of $[0,1]$. We will use $F_i:[0,1] \to \reals$ to 
%denote the integral of the valuation function of agent $i$, hence $F_i(x)$ denotes the value
%the agent has for the piece $[0,x]$, for every  $x \in [0,1]$. Note that $F_i$ is always
%continuous on $[0,1]$, even in the case where $f_i$ is not. The goal is to cut the item into 
%pieces and create a partition for the pieces such that every agent has the same value for every
%subset of the partition. For every natural $k$, a $k$-cut is a tuple 
%$(t_1, t_2, \dots, t_k) \in [0,1]^k$ such that $t_1 \leq t_2 \leq \ldots \leq t_k$. We 
%will assume $t_0 = 0$ and $t_{k+1}=1$. A $k$-cut cuts the object to at most $k+1$ pieces,
%where the $i$-th piece is $X_i := [t_{i-1}, t_i]$. A partition-vector is  
%$(s_1, s_2, \ldots, s_{k+1}) \in \{-1,+1\}^{k+1}$. We define $A_- := \{X_i: s_i =-1 \}$ and
%$A_+:= \{X_i: s_i =+1 \}$. For $y \in \{-,+ \}$, we denote $F_i(A_y)$ the valuation of agent
%$i$ for the collection of pieces in $A_y$; formally $F_i(A_y) := \sum_{j:s_j=y} \big( F_i(t_j)-
%F_i(t_{j-1}) \big)$. 

\begin{tcolorbox}[title = {$(n,n)$-\CH}]
	%{\bf (n,n)-\CH} \\
	{\bf Input:} For every agent $i \in [n]$, an arithmetic circuit $F_i$ computing
	the integral of agent $i$'s valuation function. \\
	{\bf Task:} Find an $n$-cut for $A$ such that $F_i(A_{+})=F_i(A_{-})$, for every agent $i \in [n]$.
\end{tcolorbox}

For $k < n$ a solution to the problem may or may not exist. So we define the
following decision variant of the problem. 
\begin{tcolorbox}[title = {$(n,k)$-\CH}]
	%{\bf (n, k)-\CH} \\
	{\bf Input:} For every agent $i \in [n]$, an arithmetic circuit $F_i$ computing
	the integral of agent $i$'s valuation function. \\
	{\bf Task:} Decide whether there exists a $k$-cut for $A$ such that $F_i(A_{+})=F_i(A_{-})$, for every agent $i \in [n]$.
\end{tcolorbox}

For either of these two problems, if all of the inputs are represented by
linear arithmetic circuits, then we refer to the problem as \linCH. We note that
the known hardness results~\cite{FFGZ18,FG18} for Consensus Halving fall into
this class. Specifically, those results produce valuations that are piecewise
constant, and so the integral of these functions is piecewise linear, and these
functions can be written down as linear arithmetic circuits~\cite{O02}.

\section{The Class \BU}\label{sec:class_bu}

The Borsuk-Ulam theorem states that every continuous function from the
surface of an $(d+1)$-dimensional sphere to the $d$-dimensional Euclidean space maps at least one pair of antipodal points to the same point. 

\begin{theorem}[Borsuk-Ulam]
	\label{thm:borsuk-ulam}
	Let $f : S^d \to \reals^d$ be a continuous function, where $S^d$ is a
	$(d+1)$-dimensional sphere. Then, there exists an $x \in S^d$ 
	such that $f(x)=f(-x)$.
\end{theorem}

This theorem actually works for any domain $D$ that 
is an antipode-preserving homeomorphism of $S^d$, where by ``antipode-preserving''
we mean that for every $x \in D$ we have that $-x \in D$. 
In this work, we choose $S^d$ to be 
the sphere in $d+1$ dimensions with respect to $L_1$ norm:
\begin{align*}
	S^d := \left \{ x\ |\ x=(x_1, x_2, \dots, x_{d+1}),\ \sum_{i=1}^{d+1}|x_i| = 1 
	\right \}.
\end{align*}
%Note that, perhaps confusingly, $S^d \subseteq R^{d+1}$.

We define the \emph{Borsuk-Ulam} problem as follows.
\begin{tcolorbox}[title= \borsuk]
	{\bf Input:} A continuous function $f : \reals^{d+1} \to \reals^d$ presented as
	an arithmetic circuit.\\
	{\bf Task:} Find an $x \in S^d$ such that $f(x)=f(-x)$.
\end{tcolorbox}
Note that we cannot constrain an arithmetic circuit to only take inputs from the
domain $S^d$, so we instead put the constraint that $x \in S^d$ onto the
solution.

The complexity class \BU is defined as follows.
%
%Using \borsuk as the typical problem, we define the complexity class \BU.
%Now we can define the typical problem of the complexity class BU. This class captures the complexity of search problems whose solution can be proven via the Borsuk-Ulam theorem.
\begin{definition}[\BU]
	\label{def:BU}
	The complexity class \BU consists of all search problems that can be 
	reduced to \borsuk in polynomial time under a reduction of the type described in Section \ref{sec:ar_cir_red}. 
\end{definition}

\subsection{\linBU}

When the input to a \borsuk instance is a linear arithmetic circuit, then we
call the problem \linborsuk, and we define the class \linBU as follows.

%Let us also define a subclass of \BU, called \linBU. The typical problem 
%of \linBU is \linborsuk; a restricted version of \borsuk.
%%the function $b$ given as an input is multilinear. 
%An instance of \linborsuk is given by an arithmetic circuit that does 
%not use $G_*$ gates which takes two inputs (variables in general), but is allowed to use $G_{*\zeta}(v_{in},v_{out},\zeta)$ 
%gates, where $x[v_{out}] = \zeta \cdot x[v_{in}]$ and $\zeta \in \reals$.
%\argy{Themis said something about multilinear functions. Check this.}\themis{In [Constant Rank Bimatrix Games are PPAD-Hard, Mehta] the $\zeta$ is defined to be a rational. I guess the same holds here. Please confirm and modify this in the whole paper if that is the case.}
\begin{definition}[\linBU]
	\label{def:linBU}
	The complexity class \linBU consists of all search problems that can be 
	reduced to \linborsuk in polynomial time. 
\end{definition}

We will show that $\linBU = \PPA$. The proof that $\linBU
\subseteq \PPA$ is similar to the proof that 
Etessami and Yannakakis used to show that $\linFIXP \subseteq
\PPAD$~\cite{EY10}, while the fact that 
$\PPA \subseteq \linBU$ will follow from our results on Consensus Halving in
Section~\ref{sec:bu}.

To prove $\linBU \subseteq \PPA$ we will reduce to the \emph{approximate}
Borsuk-Ulam problem.
It is well known that the Borsuk-Ulam theorem can be proved via Tucker's lemma,
and Papadimitriou noted that this implies that finding an approximate solution
to a Borsuk-Ulam problem lies in \PPA~\cite{papadimitriou1994complexity}. This
is indeed correct, but the proof provided in~\cite{papadimitriou1994complexity}
is for a slightly different problem\footnote{The problem used in
	\cite{papadimitriou1994complexity} presents the function as a polynomial-time
	Turing machine rather than an arithmetic circuit, and the Lipschitzness of the
	function is guaranteed by constraining the values that it can take.}. Since our
results will depend on this fact, we provide our own definition and
self-contained proof here.
We define the approximate Borsuk-Ulam problem as follows.

\begin{tcolorbox}[title= \epsborsuk]
	{\bf Input:} 
	A continuous function $f: \reals^{d+1} \to \reals^d$
	presented as an arithmetic circuit, along with two constants $\eps, \lambda \in \reals$.
	
	{\bf Task:} Find one of the following. 
	\begin{enumerate}
		\item A point $x \in S^d$ such that $\|f(x) - f(-x)\|_\infty \le \eps$.
		\item Two points $x, y \in S^d$ such that $\|f(x) - f(y)\|_\infty > \lambda
		\cdot \|x - y\|_\infty$. 
	\end{enumerate}
\end{tcolorbox}

The first type of solution is an approximate solution to the Borsuk-Ulam
problem, while the second type of solution consists of any two points that
witness that the function is not $\lambda$-Lipschitz continuous in the
$L_\infty$-norm. The second type of solution is necessary, because an 
arithmetic circuit is capable, through repeated squaring, of computing
doubly-exponentially large numbers, and the reduction to \tucker may not be able
to find an approximate solution for such circuits. Note however that later in Lemma \ref{lem:linborsuk-in-ppa} we reduce \linborsuk to \epsborsuk, where the former has as input a linear arithmetic circuit and therefore Lipschitzness in the \epsborsuk we reduce to is guaranteed. In other words, for the purposes of this section it would suffice to assume Lipschitzness of the input function of \epsborsuk, but here we state a more general version of the problem and show that it is also included in \PPA.
%We
%must allow violations of $\lambda$-Lipschitz continuity as solutions, because
%the size of the Tucker instance produced by the reduction of Papadimitriou is polynomial in the bit-length of
%$\lambda$. Arithmetic circuits are capable of computing doubly-exponential
%functions via repeated squaring, and for such circuits the reduction to Tucker
%
We now re-prove the result of Papadimitriou in the following lemma.

\begin{lemma}[\cite{papadimitriou1994complexity}]\label{lem:eps-borsuk-PPA}
	\epsborsuk is in \PPA.
\end{lemma}

\begin{proof}
	This proof is essentially identical to the one given by 
	Papadimitriou, but various minor changes must be made due to the fact that our
	input is an arithmetic circuit, and our domain is the $L_1$-sphere.
	His proof works by reducing to the \tucker problem. In this problem we have a
	antipodally symmetric triangulation of $S^d$ with set of vertices $V$, and a
	labelling function $L: V \to \{-1, 1, -2, 2, \dots, -d, d\}$ that satisfies
	$L(v) = -L(-v)$ for all $v \in V$. The task is to find two adjacent vertices $v$
	and $u$ such that $L(v) = -L(u)$, whose existence is guaranteed via Tucker's
	lemma. Papadimitriou's containment proof goes via the hypercube, but
	in~\cite{FFGZ18} it is pointed out that this problem also lies in PPA when the
	domain is the $L_1$-sphere $S^d$.
	
	To reduce the \epsborsuk problem for $(f, \eps, \lambda)$ to \tucker, we choose an arbitrary triangulation of $S^n$ such that the distance between
	any two adjacent vertices is at most $\eps/\lambda$. Let $g(x) = f(x) - f(-x)$.
	To determine the label of a vertex $v \in V$, first find the coordinate $i$ that
	maximises $|g(v)_i|$ breaking ties arbitrarily, and then set $L(v) = i$ if $g(v)_i > 0$ and $L(v) = -i$
	otherwise. 
	
	Tucker's lemma will give us two adjacent vertices $v$ and $u$ satisfying 
	$L(v) = -L(u)$, and we must translate this to a solution to \epsborsuk. 
	If $\| g(u) - g(v) \|_\infty > \lambda \cdot \| u - v \|_\infty$, then we have a violation of
	Lipschitz continuity. Otherwise, we have 
	\begin{align*}
		\| g(u) - g(v) \|_\infty &\le \lambda \cdot \| u - v \|_\infty \\
		&\le \lambda \cdot \frac{\eps}{\lambda}\\ 
		&\le \eps 
	\end{align*}
	Let $i =
	L(v)$. Note that by definition we have that $|g(v)_j| \le |g(v)_i|$ for all $j$,
	that $|g(u)_j| \le |g(u)_i|$ for all $j$, and that that $g(u)_i$ and $g(v)_i$
	have opposite signs. These three facts, along with the fact that 
	$\| g(u) - g(v) \|_\infty \le \eps$ imply that $|g(v)_j| \le \eps$ for all $j$.
	Hence we can conclude that $\| f(v) - f(-v) \|_\infty \le \eps$ meaning that $v$ is a solution to \epsborsuk.  
\end{proof}

To show that \linBU $\subseteq$ \PPA we will provide a polynomial-time reduction
from \linborsuk to \epsborsuk. To do this, we follow closely the technique used
by Etessami and Yannakakis to show that $\linFIXP \subseteq \PPAD$~\cite{EY10}.
The idea is to make a single call to \epsborsuk to find an approximate solution
to the problem for a suitably small $\eps$, and to then round to an exact
solution by solving a linear program. To build the LP, we depend on the fact
that we have access to the linear arithmetic circuit that represents $f$.

\begin{lemma}\label{lem:linborsuk-in-ppa}
	\linborsuk is in \PPA.
\end{lemma}

\begin{proof}
	Suppose that we have a function $f$ that is represented as a linear arithmetic
	circuit. We will provide a polynomial-time reduction to \epsborsuk.

	The first step is to argue that, for all $\eps > 0$, we can make a single call
	to \epsborsuk in order to find an $\eps$-approximate solution to the problem.
	The only technicality here is that we must choose $\lambda$ so as to ensure that
	no violations of $\lambda$-Lipschitzness in the $L_\infty$-norm can be produced as a solution.
	
	Fortunately, every linear arithmetic circuit computes a $\lambda$-Lipschitz
	function where the bit-length of $\lambda$ is polynomial in the size of the
	circuit. Moreover, an upper bound on $\lambda$ can easily be computed by
	inspecting the circuit. 
	\begin{itemize}
		\item An input to the circuit has a Lipschitz constant of 1.
		\item A $+$ gate operating on two gates with
		Lipschitz constants $x$ and $y$ has a Lipschitz constant of at most $x + y$.
		\item A $*\zeta$ gate operating on a gate with Lipschitz constant $x$ has a
		Lipschitz constant of at most $|\zeta| \cdot x$. 
		\item A max or min gate operating on two gates with Lipschitz constants $x$ and
		$y$ has a Lipschitz constant of at most $\max(x, y)$.
	\end{itemize}
	The Lipschitz constant for the circuit in the $L_\infty$-norm is then the
	maximum of the Lipschitz constants of the output nodes of the circuit.
	So, for any given $\eps > 0$ that can be represented in polynomially many bits,
	we can make a single call to \epsborsuk, in order to find an $\eps$-approximate
	solution to the Borsuk-Ulam problem.
	
	The second step is to choose an appropriate value for $\eps$ so that the
	approximate solution can be rounded to an exact solution using an LP. 
	Let $g(x) = f(x) - f(-x)$. Note that $g(x)$ can also be computed by a linear
	arithmetic circuit, and that $g(x) = 0$ if and only if $f(x) = f(-x)$.
	
	We closely follow the approach of Etessami and Yannakakis~\cite{EY10}. They use
	the fact that the function computed by a linear arithmetic circuit is
	piecewise-linear, and defined by (potentially exponentially many) hyperplanes.
	They give an algorithm that, given a point $p$ in the domain of the circuit,
	computes in polynomial time the linear function (which represents the hyperplane) that defines the output of the
	circuit for $p$. Furthermore, they show that the following can be produced in polynomial
	time from the representation of the circuit and from $p$.
	\begin{itemize}
		\item A system of linear constraints $Ax \le b$ such that a point $x$ 
		satisfies the constraints only if the linear function (which represents the hyperplane) that defines the output
		of the circuit for $p$ also defines the output of the circuit for $x$.
		
		\item A linear formula $Cx + C'$ that determines the output of the circuit for
		all points that satisfy $Ax \le b$.
	\end{itemize}
	To choose $\eps$, the following procedure is used. Let $n$ be the number of
	inputs to $g$, and let $m$ be an upper bound on the bit-size of the
	solution of any linear system with $n+1$ equations where the coefficients are drawn from
	the hyperplanes that define the function computed by $g$. This can be computed
	in polynomial time from the description of the circuit, and $m$ will have
	polynomial size in relation to the description of the circuit.
	We choose $\eps < 1/2^m$.
	
	We make one call to \epsborsuk to find a point $p \in S^n$ such that $\|f(p) -
	f(-p)\|_\infty \le \eps$, meaning that $\|g(p)\|_\infty \le \eps$. 
	The final step is to round this to an exact solution of  \borsuk. 
	To do this, we can modify the linear program used by 
	Etessami and Yannakakis~\cite{EY10}. 
	We apply the operations given above to the circuit $g$ and the point $p$ to
	obtain the system of constraints $Ax \le b$ and the formula $Cx + C'$ for the
	hyperplane defining the output of $g$ for $p$.
	We then solve the following linear program.
	The variables of the LP are a vector $x$ of length $n$,
	and a scalar $z$. The goal is to minimize $z$ subject to:
	\begin{align*}
		Ax &\le b \\
		(Cx)_i + C'_i &\le z & \text{for $i = 1, \dots, n$} \\
		-((Cx)_i + C'_i) &\le z & \text{for $i = 1, \dots, n$} \\
		x_i &\ge 0 & \text{for each $i$ with $p_i \ge 0$} \\
		x_i &\leq 0 & \text{for each $i$ with $p_i < 0$} \\
		\sum_{i = 1}^n |x_i| &= 1 & \text{(see below regarding $|x_i|$)}
	\end{align*}
	The first constraint ensures that we remain on the same cell as the one
	defining the output of $g$ for $p$. The second and third constraints ensure
	that $\|g(x)\|_{\infty} \le z$. The fourth and fifth constraints ensure that $x_i$ has the
	same sign as $p_i$, while the sixth constraint ensures that $x$ lies on the
	surface $S^n$. Note that the $|x_i|$ operation in the sixth constraint is not a
	problem, since the fourth and fifth constraints mean that we know the sign of
	$x_i$ up front, and so we just need to add
	either $x_i$ or $-x_i$ to the sum. All of the
	above implies that that $x$ is a $z$-approximate solution of \borsuk for $f$. 
	
	We must now argue that the solution sets $z = 0$. First we note that the LP has
	a solution, because the point $(p, \eps)$ is feasible, and the LP is
	not unbounded since $z$ cannot be less than
	zero due to the second and third constraints.
	So let $(x^*, z^*)$ be an optimal solution. This solution lies at the
	intersection of $n+1$ linear constraints defined by rationals drawn
	from the circuit representing $g$, and so it follows that $z^*$
	is a rational of bit length at most $m$. Since $0 \le z^* \le \eps < 1/2^m$, it
	follows that $z^* = 0$, and thus $x^*$ is an exact solution to \borsuk for $f$.    
\end{proof}
%
%The proofs of these two lemmata can be found in Appendix \ref{app:eps-borsuk-PPA} and \ref{app:linborsuk-in-ppa} respectively.

\section{Containment Results for \CH}
\label{sec:bu}

%%%%%%%%%%%%%%%%%%%%%%%%

%In this section we show containment results about consensus halving problems. 
%We prove that $(n,n)$-\CH is in \BU when $F_i$s are given by general arithmetic circuits 
%and that it belongs in \linBU when $F_i$s are given by linear circuits. In addition, we prove 
%that the decision the problem, $(n,k)$-\CH, belongs in \etr for every $k$.

\subsection{$(n, n)$-\CH is in \BU and \linBU = \PPA}
\label{sec:CHinBU}

We show that $(n,n)$-\CH is contained in \BU. Simmons and Su \cite{SS03} show the
existence of an $n$-cut solution to the Consensus Halving problem by applying
the Borsuk-Ulam theorem, and we follow their approach in this reduction.
However, we must show that the approach can be implemented using arithmetic
circuits. We take care in the reduction to avoid $G_{*}$ gates, and so if the
inputs to the problem are all linear arithmetic circuits, then our reduction
will produce a \linborsuk instance. Hence, we also show that $(n, n)$-\linCH is
in \linBU. 
%The proof can be found in Appendix~\ref{app:CHinBU}.

\begin{theorem}\label{thm:CHinBU}
	The following two containments hold.
	\begin{itemize}
		\item $(n,n)$-\CH is in \BU. 
		\item $(n, n)$-\linCH is in \linBU.
	\end{itemize}
\end{theorem}

\begin{proof}
	Let us first summarise the approach used by 
	Simmons and Su \cite{SS03}. Given valuation functions $F_i$ for the $n$ agents,
	they construct a Borsuk-Ulam instance given by a function $b : S^n \rightarrow
	\reals^n$. Each point $(x_1, x_2, \dots, x_{n+1}) \in S^n$ can be interpreted as
	an $n$-cut of $[0, 1]$, where $|x_i|$ gives the \emph{width} of the $i$th piece,
	and the sign of $x_i$ indicates whether the $i$th piece should belong in $A_+$
	or $A_-$. They then define $b(x)_i = F_i(A_+)$ for each agent $i$. The fact that
	$-x$ flips the sign of each piece, but not the width, implies that $b(-x)_i =
	F_i(A_-)$. Hence, any point that satisfies $b(x) = b(-x)$ has the property that 
	$F_i(A_+) = F_i(A_-)$ for all agents $i$, and so is a solution to Consensus Halving.
	
	Our task is to implement this reduction using arithmetic circuits and construct in polynomial time a \borsuk instance from a $(n,n)$-\CH instance. Suppose that
	we are given arithmetic circuits $F_i$ implementing the integral of each agent's
	valuation function. We show how to map each $F_i$ to a function
	$b(x)_i = F_i(A_+)$ computable via a linear arithmetic circuit, where $x \in S^n$, i.e. a \borsuk instance. The tricky
	part of this, is that for each agent $i$ we must include the $j$-th piece in the sum if and only if $x_j$
	is positive. Then we show how to map a solution $x$ back to a solution of $(n,n)$-\CH.

	We begin by observing that the operation of $|x|$ can be implemented via a
	linear arithmetic circuit. Specifically, via the following construction:
	\begin{equation*}
		|x| := \max(x, 0) + \max(-x, 0).
	\end{equation*}
	Hence, we can implement $|x|$ using only gates $G_{\max}$, $G_{+}$, and $G_{\zeta}$.
	Then, we define $t_0 := 0$, and for each $j$ in the range $1 \le j  \le n+1$, define:
	\begin{equation}\label{eq:cut_points}
		t_{j}:= t_{j-1} +  |x_{j}|.
	\end{equation}
	The value of $t_j$ gives
	the end of the $j$-th piece, and note that $t_{n+1} = 1$.
	Next, for each $j$ in the range $1 \le j \le n+1$ we define:
	\begin{equation*}
		p_{j}:= \max(x_j, 0).
	\end{equation*}
	Note that $p_j$ is $x_j$ whenever $x_j$ is positive, and zero otherwise.
	Finally, for $1 \le j \le n+1$ define:
	\begin{equation*}
		q_{j}:= F_i(t_{j-1} + p_j) - F_i(t_{j-1}).
	\end{equation*}
	Using the reasoning above, we can see that $q_j$ is agent $i$'s valuation for
	piece $j$ whenever $x_j$ is positive, and zero otherwise. So we can define
	\begin{equation*}
		b(x)_i = \sum_{j=1}^{n+1} q_j, 
	\end{equation*}
	implying that $b(x)_i = F_i(A_+)$, as required.
	
	Finally, we need to map a solution $x$ of \borsuk to a pair of $(n,n)$-\CH solutions, i.e. a vector of cut-points $(t_1, t_2, \dots, t_n)$. Recall that the cut points correspond to a pair of symmetric solutions where, in each, the signs of the resulting pieces are alternating (by definition) and the two solutions have opposite signs. 

	Let the input to $(n,n)$-\CH be a general circuit with gates $G_{\zeta}, G_{+}, G_{-}, G_{*\zeta}, G_{*}, G_{\max}, G_{\min}$, or a linear circuit, where gate $G_{*}$ is disallowed.
	From a solution of \borsuk we map back to a solution of $(n,n)$-\CH by constructing a circuit in which we allow the use of an extra {\em comparison} gate $G_{>}$. This gate takes an input $v_{in}$ and outputs $1$ if $v_{in} > 0$ and $0$ otherwise. One can see that the function this gate implements is discontinuous only for $v_{in} = 0$, contrary to the rest of the gates that implement continuous functions. This fact, however, does not affect the validity of the reduction (see also Section \ref{sec:ar_cir_red}). Whether a mapping can be constructed without the use of $G_{>}$, or in general, with using only gates that implement continuous functions is left as an open problem.
	
	%	 From a solution of \borsuk we map back to a solution of $(n,n)$-\CH by constructing a circuit using the additional gate $G_{>}$. 
	We denote the operation this gate implements in the following way: $\{ x \}_{>} := 1$, if $x > 0$, and $\{ x \}_{>} := 0$, otherwise. This circuit computes the $n$-cut that is a solution to $(n,n)$-\CH.
	Given a solution of \borsuk, i.e. a vector $(x_1, x_2, \dots, x_{n+1}) \in S^n$, we construct a circuit that has two stages: (i) first it shifts all $x_{j} = 0$ to position $n+1$ of the vector, (ii) then for every two consecutive $x_{j}$'s it merges them if they have the same sign, thus resulting to a vector $x$ with coordinates of alternating sign and consecutive zeros at the rightmost positions. One should note that merging a pair of consecutive coordinates of $x$ that have the same sign, and transferring all coordinates of value zero at the rightmost position of $x$ while maintaining the order of the rest of the coordinates, does not affect the Consensus Halving solution. That is because by such operations, the positive and negative intervals remain the same; only cuts between two consecutive pieces or cuts that are more than one on the same position are transferred to position $1$ of the interval $[0,1]$. In other words, the positive and negative valuations $F_{i}(A_{+})$ and $F_{i}(A_{-})$ remain the same after such operations since the positive and negative intervals remain at the same positions on $[0,1]$. What we want is to bring the Consensus Halving solution into the form of a valid $(n,n)$-\CH solution, i.e. an $n$-cut whose pieces have alternating signs.
	
	For the implementation of checking whether a coordinate $x_{j}$ is zero and shifting it one position to the right, we make the following construction:
	\begin{align*}
		x_{j} &= \left(\{x_{j}\}_{>} + \{-x_{j}\}_{>}\right)*x_{j} + \left(1 - \{x_{j}\}_{>} - \{-x_{j}\}_{>}\right)*x_{j+1}, \\
		x_{j+1} &= \left(\{x_{j}\}_{>} + \{-x_{j}\}_{>}\right)*x_{j+1} + \left(1 - \{x_{j}\}_{>} - \{-x_{j}\}_{>}\right)*x_{j}.
	\end{align*}
	Therefore, to move a zero (if it exists) to the rightmost position of $x$ we need to implement the above two functions for every $j \in [n]$ in increasing order, meaning that after we implement the circuit for the pair $x_1$, $x_2$ we then implement $x_2$, $x_3$ and so on until $x_n$, $x_{n+1}$. To implement stage (i) we have to iterate this procedure $n$ times, since in the worst case there will be $n$ coordinates with value zero in $x$. Note that for stage (i) we need no more than $O(n^2)$ gates.
	
	Now in our vector, starting from the left, there are consecutive non-zero values and after them consecutive zero values. What remains to be done is to implement stage (ii), i.e. to merge pairs of consecutive coordinates with the same sign. To do that for a pair $x_{j}$, $x_{j+1}$, we make the following construction:
	\begin{align*}
		x_{j} &= x_{j} + \{x_{j} * x_{j+1}\}_{>}*x_{j+1}, \\
		x_{j+1} &= \left(1 - \{x_{j} * x_{j+1}\}_{>}\right)*x_{j+1}.
	\end{align*}
	Therefore, either there will be no change in the values of $x_j$ and $x_{j+1}$ if they are of opposite sign, or $x_{j}$ becomes $x_{j}+x_{j+1}$ and $x_{j+1}$ becomes zero. In the latter case we will have introduced a zero to the vector. That is why right after the above construction for some $j \in [n]$ we implement a shifting of the (possibly introduced) zero to the rightmost position of $x$ using the aforementioned procedure of stage (i). We do this for every $j \in [n]$ in an increasing order. Note that for stage (ii) we need no more than $O(n^2)$ gates.
	
	After implementing the aforementioned two stages, the resulting vector $x = (x_1, x_2, \dots, x_{n+1})$, starting from the left, has coordinates of alternating sign and at its rightmost positions it has zeros. Finally, we compute the $n$-cut $(t_1, t_2, \dots, t_n)$ in a straight-forward way using equation \eqref{eq:cut_points}, where we initialize $t_0 := 0$. Also, always $t_{n+1} = 1$ and it is discarded. Note that for the above constructions no more than $O(n^2)$ gates were needed in total. Since the pieces of these cuts have alternating sign, this is a valid solution to $(n,n)$-\CH and the proof is complete.   
	%	 
	%	To complete the proof, it suffices to note that none of the operations specified
	%	above use the gate $G_{*}$, and so if each $F_i$ is specified by a linear
	%	arithmetic circuit, then $b$ will also be a linear arithmetic circuit.  
\end{proof}

Theorem \ref{thm:CHinBU} also implies that $\PPA \subseteq \linBU$, thereby completing
the proof that $\PPA = \linBU$. Specifically, Filos{-}Ratsikas and Goldberg have
shown that \emph{approximate}-$(n, n)$-\CH is PPA-complete, and their valuation functions
are piecewise constant \cite{FG18}. Therefore, the integrals of these functions are
piecewise linear, and so their approximate-$(n, n)$-\CH instances can be reduced to
$(n, n)$-\textsc{Linear} \CH. Hence $(n, n)$-\linCH is \PPA-hard, which along with
Lemma \ref{lem:linborsuk-in-ppa} implies the following corollary.

\begin{corollary}
	\label{cor:PPA-linBU}
	$\PPA = \linBU$.
\end{corollary}

\subsection{$(n,k)$-\CH is in \etr}\label{sec:dch-in-etr}

The existential theory of the reals consists of all true existentially
quantified formulae using the connectives $\{\land, \lor, \lnot\}$ over
polynomials compared with the operators $\{<, \le, =, \ge, >\}$. The complexity
class \etr captures all problems that can be reduced in polynomial time to
the existential theory of the reals.

%The complexity class \etr captures problems of the {\em existential theory of
%the reals}, which 
%it asks whether a system of polynomials over $\reals^n$ has a solution. 

We prove that $(n,k)$-\CH is in $\etr$. The reduction simply
encodes the arithmetic circuits using ETR formulas, and then constrains $F_i(A_+) =
F_i(A_-)$ for every agent $i$. 
%The proof can be found in Appendix \ref{app:CHinETR}.

\begin{theorem}
	\label{thm:CHinETR}
	$(n,k)$-\CH is in \etr.
\end{theorem}

\begin{proof}
	
	The first step is to argue that an arithmetic circuit can be implemented as an
	ETR formula.
	%For each agent 
	%$i$ we use the arithmetic circuit that gives his valuation to create a 
	%set of polynomial equations that capture his valuation function. Then, we 
	%add further constraints to our \etr formula in order to express the extra 
	%conditions that give a correct solution with $n-1$ cuts for the Consensus 
	%Halving instance. 
	Let $(V, \mathcal{T})$ be the arithmetic circuit. For every vertex $v \in V$ we
	introduce a new variable $x_v$. For every gate $g \in \mathcal{T}$ we introduce
	a constraint. For the gates in the set $\{G_{\zeta}, G_{+}, G_{-}, G_{*\zeta},
	G_{*}\}$ the constraints simply implement the gate directly, eg., for a gate
	$G_{+}(v_{\text{in1}}, v_{\text{in2}}, v_{\text{out}})$ we use the constraint
	$x[v_{\text{out}}] =
	x[v_{\text{in1}}] + x[v_{\text{in2}}]$. For a gate $G_{\max}(v_{\text{in1}},
	v_{\text{in2}}, v_{\text{out}})$ we
	use the formula
	\begin{equation*}
		\big((x[v_{\text{out}}] = x[v_{\text{in1}}]) \land (x[v_{\text{in1}}] \geq
		x[v_{\text{in2}}])\big) \lor
		\big((x[v_{\text{out}}] = x[v_{\text{in2}}]) \land (x[v_{\text{in2}}] \geq
		x[v_{\text{in1}}])\big),
	\end{equation*}
	and likewise for a gate $G_{\min}(v_{in1}, v_{in2}, v_{out})$ we use the formula
	\begin{equation*}
		\big((x[v_{\text{out}}] = x[v_{\text{in1}}]) \land (x[v_{\text{in1}}] \leq
		x[v_{\text{in2}}])\big) \lor
		\big((x[v_{\text{out}}] = x[v_{\text{in2}}]) \land (x[v_{\text{in2}}] \leq
		x[v_{\text{in1}}])\big). 
	\end{equation*}
	Taking the conjunction $C$ of the constraints for each of the gates yields an ETR
	formula that implements the circuit.
	
	Now we perform the reduction from Consensus Halving to the existential theory of
	the reals. Suppose that we have been given, for each agent $i$, an arithmetic circuit $F_i$
	implementing the integral of agent $i$'s valuation function. We have already
	shown in the proof of Theorem~\ref{thm:CHinBU} that, given a description of
	a $k$-cut given as a point in $S^k$, we can 
	create a circuit implementing $F_i(A_+)$ and a
	circuit implementing $F_i(A_-)$ for each agent $i$.
	We also argued in that
	proof that $\sum_{j = 1}^{k+1} |x_j|$ can be implemented as an arithmetic circuit.
	Our ETR formula is as follows.
	\begin{equation*}
		\exists x \cdot \left( \bigwedge\limits_{i=1}^n F_i(A_+) = F_i(A_-) \right) \land C \land 
		\left(\sum_{j = 1}^{k+1} |x_j| = 1 \right).
	\end{equation*}
	The first set of constraints ensure that $x$ is a solution to the Consensus Halving problem, the second one implements as showed above the $\max$ and $\min$ operations that are not allowed in an ETR formula, and the final constraint ensures that $x \in S^n$.     
\end{proof}

Using the same technique, we can also reduce $\borsuk$ to an \etr formula. In
this case, we get an ETR formula that always has a solution. Let us define here the class \fetr (\texttt{Function \etr}) which contains all search problems whose corresponding decision version lies in \etr. We also define the class \tfetr (\texttt{Total Function \etr}) as the subclass of \fetr which contains the search problems whose decision version outputs always ``yes''. As \etr is the analogue of \NP, \fetr and \tfetr are the analogues of \FNP and \TFNP respectively in the Blum-Shub-Smale computation model \cite{BSS89}. 

\begin{theorem}
	\label{thm:butfetr}
	$\BU \subseteq \tfetr$.
\end{theorem}

\begin{proof}
	The proof is essentially identical to the proof of Theorem~\ref{thm:CHinETR},
	and the only difference is that instead of starting with a Consensus Halving
	instance, we start with an arbitrary arithmetic circuit representing the
	function $f : S^d \to \reals^d$, for which we wish to find a point $x$
	satisfying $f(x) = f(-x)$. We implement the arithmetic circuit in the same
	way as in Theorem~\ref{thm:CHinETR}, and our ETR formula is:
	\begin{equation*}
		\exists x \cdot \left( \bigwedge\limits_{i=1}^d f_i(x) = f_i(-x) \right) \land C \land 
		\left(\sum_{j = 1}^{d+1} |x_j| = 1 \right),
	\end{equation*}
	where $C$ is the conjunction of the constraints that implement $\max$ and $\min$ gates.
\end{proof}

\section{Hardness Results for \CH}\label{sec:hardness_res}

In this section we give an overview of our hardness results for Consensus Halving. 
Full proofs will be given in subsequent sections. We prove that $(n,n)$-\CH is 
\FIXP-hard and that $(n,n-1)$-\CH is \etr-hard. These two reductions share a
common step of embedding an arithmetic circuit into a Consensus Halving
instance. So we first describe this step, and then move on to proving the two
individual hardness results.
An outline of the embedding step is described in Section \ref{subsec: cir-to-ch}, which concludes to Lemma \ref{lem:circuit-embed}. The detailed proof of that lemma is presented in Section \ref{app:circuit-embed}. 

Then, in Section \ref{subsec:FIXP-hard} we present a polynomial-time reduction from the \FIXP-complete problem of computing a Nash equilibrium in a $d$-player strategic form game to the problem of computing a $(n,n)$-\CH solution. Finally, in Section \ref{sec:ETR}, after proving \etr-completeness for an auxiliary problem, we reduce from it to the $(n,n-1)$-\CH problem. The implied \etr-hardness of the latter problem, together with its \etr-membership by Theorem \ref{thm:CHinETR} proves the required \etr-completeness.

\subsection{Embedding a circuit in a \CH instance: an outline}\label{subsec: cir-to-ch}

Our approach is inspired by~\cite{FFGZ18}, who provided a reduction from
$\eps$-\gcircuit~\cite{chen2009settling,Rub18} to approximate Consensus Halving. 
However, our construction deviates significantly from theirs due to several reasons.

Firstly, the reduction in~\cite{FFGZ18} works \emph{only} for approximate
Consensus Halving. Specifically, some valuations used in that construction have
the form of $1/\eps$, where $\eps$ is the approximation guarantee, so the 
construction is not well-defined when $\eps=0$ as
it is in our case. Many of the gate gadgets used in~\cite{FFGZ18} cannot be used
due to this issue, including the $\max$ gate, which is crucially used in that
construction to ensure that intermediate values do not get too large. We provide
our own implementations of the broken gates. Our gate gadgets only work when the
inputs and outputs lie in the range $[0, 1]$, and so we must carefully construct
circuits for which this is always the case. The second major difference is that
the reduction in~\cite{FFGZ18} does not provide any method of multiplying two
variables, which is needed in our case. We construct a gadget to do this, based
on a more primitive gadget for squaring a single variable.

%The construction of our
%reduction is inspired by~\cite{FFGZ18}, but it deviates significantly from
%theirs for various reasons. Firstly, the gates used in~\cite{FFGZ18} by default
%output values in $[0,1]$; this is because they use the $\eps$-\gcircuit for
%their \PPAD-hardness result and a Boolean circuit for their \NP-hardness result.
%Furthermore, our circuits contain multiplication gates, hence we have to create
%completely new gadgets  to implement such gates. In addition, in~\cite{FFGZ18}
%the authors seek for approximate solutions for \CH where a discrepancy of
%$\eps>0$ between the values of $A_+$ and $A_-$ is allowed for every agent, i.e
%$|F_i (A_+) - F_i(A_-)| \leq \eps$ for every $i$. For this reason their gadgets
%where parameterised by $\eps$. Specifically, there were contiguous segments of
%the object of length $\eps$  where agents had value of $\frac{1}{\eps}$. Hence,
%these gadgets could not be used when $\eps=0$ and we constructed ours from
%scratch.

\subsubsection{Special circuit}\label{sec:special_cir}

Our reduction from an arithmetic circuit to Consensus Halving will use a very
particular subset of gates. Specifically, we will not use $G_{\text{min}}$,
$G_{\text{max}}$, or $G_{*}$, and we will restrict $G_{*\zeta}$ so that
$\zeta$ must lie in $(0, 1]$. We do however introduce three new gates, shown in 
Table \ref{tbl:special_gates}. The gate $G_{()^2}$ squares its input, the gate
$G^{[0, 1]}_{*2}$ multiplies its input by two, but requires that the input be in
$[0, 1/2]$, and the gate $G^{[0, 1]}_{-}$ is a special minus gate that takes as inputs $a,b \in [0,1]$ and outputs $\max \{a-b , 0\}$.
\begin{table}
	\begin{center}
		\begin{tabular}{|l|l|l|}
			\hline
			\textbf{Special Gate}                                & \textbf{Constraint}    & \textbf{Ranges} \\ \hline
			\quad $G_{()^2}(v_{in}, v_{out})$ \quad  &  \quad $x[v_{out}] = (x[v_{in}])^2$ \quad &  \quad $x[v_{in}] \in [0,1]$ \quad           \\ \hline
			\quad $G_{*2}^{[0,1]}(v_{in}, v_{out})$ \quad  &  \quad $x[v_{out}] = x[v_{in}] \cdot 2 $ \quad  &  \quad $x[v_{in}] \in [0,1/2]$ \quad            \\ \hline
			\quad $G_{-}^{[0,1]}(v_{in1}, v_{in2}, v_{out})$ \quad  &   \quad $x[v_{out}] = \max \{ x[v_{in1}] - x[v_{in2}] , 0\} $ \quad  &  \quad $x[v_{in1}], x[v_{in2}] \in [0,1]$ \quad          \\ \hline
		\end{tabular}
	\end{center}
	\caption{The special types of gates, their constraints and ranges of input.}
	\label{tbl:special_gates}
\end{table}

We note that $G_{\text{min}}$, $G_{\text{max}}$, and $G_{*}$ can be implemented
in terms of our new gates according to the following identities.
\begin{align*}
	\max\{a, b\} &= \frac{a+b}{2} + \frac{|a-b|}{2} = \frac{a}{2} + \frac{b}{2}  + \frac{1}{2} 
	\max\{a - b, 0\} + \frac{1}{2}  \max\{b - a, 0\}, \\
	\min\{a, b\} &= \frac{a+b}{2} - \frac{|a-b|}{2} = \frac{a}{2} + \frac{b}{2}  - \frac{1}{2} 
	\max\{a - b, 0\} - \frac{1}{2}  \max\{b - a, 0\} , \\
	a\cdot b &= 2  \left[ \left( \frac{a}{2}   +  \frac{b}{2}  \right)^2 - \left( \left( \frac{a}{2}   \right)^2 + \left( \frac{b}{2}  \right)^2 \right) \right].
\end{align*}

Also, a very important requirement of the special circuit is that both inputs of any $G_{+}$ gate are in $[0,1/2]$. To make sure of that, we downscale the inputs before reaching the gate, and upscale the output, using the fact that $a+b = (a/2 + b/2) \cdot 2$.

\subsubsection{The reduction to \CH}

\begin{figure}
	\begin{center}
		\includegraphics[scale=0.45]{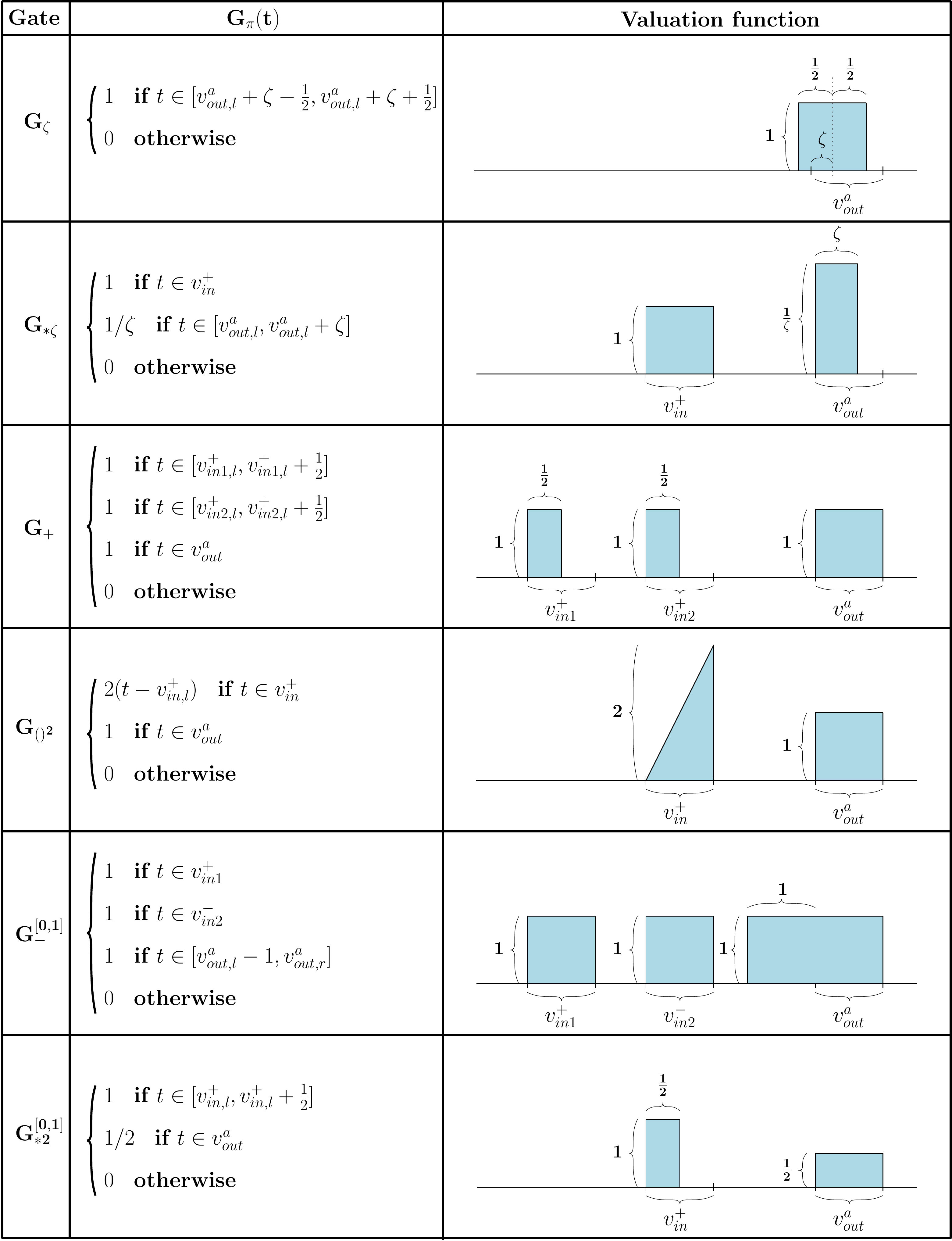}
	\end{center}
	\caption{Gates and their corresponding functions $G_{\pi}(t)$.}\label{gates-table}
\end{figure}

The reduction follows the general outline of the reduction given
in~\cite{FFGZ18}. The construction is quite involved, and so we focus on the
high-level picture here.

Each gate is implemented by 4 agents, namely $ad, mid, cen, ex$ in the Consensus 
Halving instance. The values computed by the gates are encoded by the
positions of the cuts that are required in order to satisfy these agents.
Agent $ad$ performs the exact mathematical operation of the gate, and feeds the
outcome in $mid$, who ``trims'' it in accordance with the gate's actual
operation. Then $mid$ feeds her outcome to $cen$ and $ex$, who make a copy
of $mid$'s correct value of the gate, with ``negative'' and ``positive'' labels
respectively. This value with the appropriate label will be input to other
gates.

The most important agents are the ones that perform the mathematical operation of each gate, i.e. agents $ad$.
Figure \ref{gates-table} shows the part of the valuation functions of these
agents that perform the operation. Each
figure shows a valuation function for one of the agents, meaning that the blue
regions represent portions of the object that the agent desires. The agent's
valuation for any particular interval is the integral of this function
over that interval. 

To understand the high-level picture of the construction, let us look at the
construction for $G_{*\zeta}$. The precise valuation functions of the agents in the construction (see \eqref{val-functions}) ensure that
there is exactly one \emph{input} cut in the region $v_{in}^{+}$. The leftmost
piece due to that cut in that region will belong to $A_+$, while the rightmost will belong to
$A_-$. It is also ensured that there is exactly one \emph{output} cut in the
region $v_{out}^{a}$, and that  the first piece in that region will belong to
$A_-$ and the second will belong to $A_+$. 

Suppose that gate $g_i$ in the circuit is of type $G_{*\zeta}$ and we want to implement it through a \CH instance. If we treat $v_{in}^{+}$ and $v_{out}^{a}$ in Figure \ref{gates-table} as representing $[0, 1]$, then agent $ad_i$ will take as input a cut at point $x \in v_{in}^{+}$. In order to be satisfied, $ad_i$ will impose a cut at point $y \in v_{out}^{a}$, such that $F_i(A_+) = F_i(A_-)$, where:
%
%\begin{align*}
$F_i(A_+) = x + (\zeta - y)/\zeta$ and 
%\\
$F_i(A_-) = (1 - x) + y/\zeta$.
%\end{align*}
Simple algebraic manipulation can be used to show that $ad_i$ is satisfied only when $y = \zeta\cdot x$, as required.

We show that the same property holds for each of the gates in
Figure~\ref{gates-table}. Two notable constructions are for the gates $G_{()^2}$
and $G^{[0, 1]}_{-}$. For the gate $G_{()^2}$ the valuation function of agent $ad$ is
non-constant, which is needed to implement the non-linear squaring function.
For the gate $G^{[0, 1]}_{-}$, note that the output region $v_{out}^{a}$ only covers half of the
possible output space. The idea is that if the result of $x[v_{in1}] - 
x[v_{in2}]$ is negative, then the output cut will lie before the output region, which will be interpreted as a zero output by agents $mid, cen, ex$ in the
construction. On the other hand, if the result is positive, the result will lie
in the usual output range, and will be interpreted as a positive number. An example where $x[v_{in1}] = 1/4$ and $x[v_{in2}] = 3/4$ is shown in Figure \ref{example_minus}.

\begin{figure}
	\begin{center}
		\includegraphics[scale=0.82]{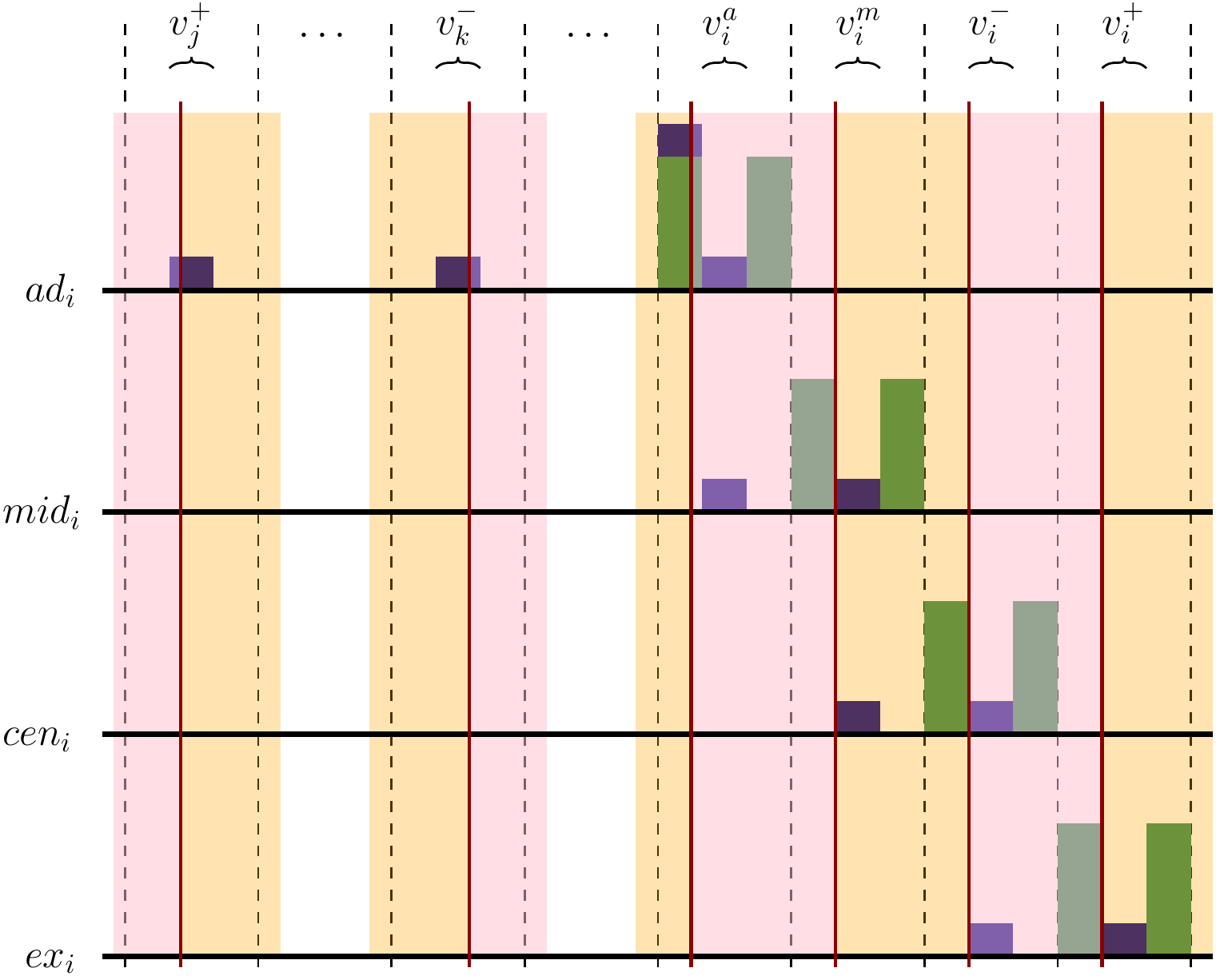}
	\end{center}
	\caption[Example: computation of gate $G_{-}^{[0,1]}$ by a \CH instance.]{An example where the computation at the output $v_{out}:=v_{i}$ of
		a $G_{-}^{[0,1]}$ gate with inputs $v_{in1}:=v_j$ and $v_{in2}:=v_k$ is
		simulated by the \CH instance. Here $x[v_{j}] = 1/4$ and $x[v_{k}] = 3/4$, hence
		$x[v_{i}] = 0$. The information about the values of the inputs is encoded by the
		cuts (red lines) in intervals $v_{j}^{+}$, and $v_{k}^{-}$ imposed by agents
		$ex_j$ and $cen_{k}$ respectively. The blue and green shapes depict the area
		below the valuation function of each of the 4 agents. The pink regions have
		label ``$+$'' while the yellow have label ``$-$''. Agent $ad_i$ performs the
		subtraction, by demanding that she is satisfied, and places a cut $1/10$ to the
		left of the left endpoint of interval $v_{i}^a$. Then agent $mid_i$ gets
		satisfied by placing a cut at exactly the left endpoint of interval $v_{i}^m$,
		thus encoding the value 0 which is the correct output value of the gate.
		Finally, agents $cen_i, ex_i$ copy this value by enforcing similar cuts at the
		left endpoints of intervals $v_{i}^{-}$ and $v_{i}^{+}$ respectively. The
		encoded values in the latter two intervals are the ``negative'' and ``positive''
		version of $x[v_{i}]$.}\label{example_minus}
\end{figure}

Ultimately, this allows us to construct
a Consensus Halving instance that implements this circuit.
This means that for any 
$x \in [0, 1]^n$, we can encode $x$ as a set of cuts, which then force 
cuts to be made at each gate gadget that encode the correct output for that
gate.
%The full details of the construction are quite involved, and so we defer them to
%Appendix \ref{app:circuit-embed}, where the following result is shown.
\begin{lemma}
	\label{lem:circuit-embed}
	Suppose that we are given an arithmetic circuit with the following properties.
	\begin{itemize}
		\item The circuit uses the gates $G_{\zeta}, G_{+}, G_{*\zeta}, G_{()^2}, G_{-}^{[0,1]}, G_{*2}^{[0,1]}$.
		\item Every $G_{\zeta}$ and $G_{* \zeta}$ has $\zeta \in \mathbb{Q} \cap (0, 1]$.
		\item For every input $x \in [0, 1]^n$, all intermediate values computed by the
		circuit lie in $[0, 1]$.
	\end{itemize}
	We can construct a Consensus Halving instance that implements this circuit.
\end{lemma}
The proof of this lemma is presented in Section \ref{app:circuit-embed}.

\subsection{$(n,n)$-\textsc{Consensus Halving} is \FIXP-hard}
\label{subsec:FIXP-hard}

We show that $(n,n)$-\CH is \FIXP-hard by reducing from the problem of finding a
Nash equilibrium in a $d$-player game, which is known to be \FIXP-complete \cite{EY10}. As shown in~\cite{EY10}, this problem can be reduced to the
Brouwer fixed point problem: given an arithmetic circuit computing a function $F
: [0, 1]^n \to [0, 1]^n$, find a point $x \in [0, 1]^n$ such that $F(x) = x$. In
a similar way to \cite{FFGZ18}, we take this circuit and embed it into a
Consensus Halving instance, with the outputs looped back to the inputs. Since
Lemma~\ref{lem:circuit-embed} implies that our implementation of the circuit is correct, this means that any solution to
the Consensus Halving problem must encode a point $x$ satisfying $F(x) = x$.

One difficulty is that we must ensure that the arithmetic circuit that
we build falls into the class permitted by Lemma~\ref{lem:circuit-embed}. To do
this, we carefully analyse the circuits produced in \cite{EY10}, and we modify
them so that all of the preconditions of Lemma~\ref{lem:circuit-embed} hold.
This gives us the following result.

\begin{theorem}\label{thm:CH_is_FIXP-hard}
	$(n,n)$-\CH is \FIXP-hard. 
\end{theorem}

The proof of this theorem is presented in Section \ref{app:CH_is_FIXP-hard}.
Theorem \ref{thm:CH_is_FIXP-hard}, together with Theorem \ref{thm:CHinBU} give the following corollary.

\begin{corollary}
	\FIXP $\subseteq$ \BU.
\end{corollary}

\subsection{$(n,n-1)$-\CH is \etr-complete}
\label{sec:ETR}

We will show the \etr-hardness of $(n,n-1)$-\CH by reducing from the following
problem \bconj, which we prove it is \etr-complete.

\begin{definition}[\bconj]
	Let $p_1, \ldots, p_k : [0,1]^n \to \reals$ be a family of polynomials, where
	each one of them is given as a sum of monomials with integer coefficients. 
	\bconj asks whether the polynomials have a common zero.
\end{definition}

Then, we reduce the above problem to the following one. 
\begin{definition}[\feas, \bfeas]
	\label{def:feas}
	Let $p(x_1, \dots, x_m)$ be a polynomial.
	%	in the standard form 
	%as a sum of monomials. 
	\feas asks whether there exists a point $(x_1, \dots, x_m) \in \reals^m $ 
	that satisfies $p(x_1, \dots, x_m) = 0$. 
	\bfeas asks whether there exists a point $(x_1, \dots, x_m) \in [0,1]^m$ that satisfies $p(x_1, \dots, x_m) = 0$.
\end{definition}

The idea is to turn the polynomial into a circuit, and then embed that circuit
into a Consensus Halving instance using Lemma~\ref{lem:circuit-embed}. 
As before, the main difficulty is ensuring that the preconditions
of Lemma~\ref{lem:circuit-embed} are satisfied. To do this, we must ensure that
the inputs to the circuit take values in $[0, 1]$, which is not the case if we reduce directly
from \feas. Instead, we first consider the problem $\bfeas$, in which
$x$ is constrained to lie in $[0, 1]^n$ rather than $\reals^n$, and we show the
following result.

\begin{lemma}\label{lem:Feas_is_etr-c}
	\bfeas is \etr-complete even for a polynomial of maximum sum of variable exponents in each monomial equal to 4.
\end{lemma}

The proof of that lemma is presented in Section \ref{app:Feas_is_etr-c}. Consequently, via a polynomial-time reduction from \bfeas to $(n,n-1)$-\CH and Theorem \ref{thm:CHinETR}, we prove the following result.
\begin{theorem}\label{thm:CH_is_etr-c}
	$(n,n-1)$-\CH is \etr-complete.
\end{theorem}
The proof of the above theorem is presented in Section \ref{app:CH_is_etr-c}.

%%%%%%%%%%%%%%%%%%%%%%%%%%%%%%%%%%%

\section{Proof of Lemma \ref{lem:circuit-embed}}\label{app:circuit-embed}

In this section, the detailed construction of a \CH instance from an arbitrary given special circuit is presented. A special circuit is an arithmetic circuit with the properties described in the statement of Lemma \ref{lem:circuit-embed} (see Section \ref{sec:special_cir} for a detailed definition). After the construction, a correspondence of circuit to \CH solutions is proven, which completes the proof of the lemma.

\subsection{Special circuit to \CH instance}\label{app:cir_to_ch}
Consider a circuit $H = (V, \tcal)$ that uses gates in $\{G_{\zeta}, G_{+}, G_{*\zeta}, G_{()^2}, G_{-}^{[0,1]}, G_{*2}^{[0,1]}\}$, with $\zeta \in \mathbb{Q} \cap (0,1]$, each gate's inputs/output are in $[0,1]$, and both inputs of $G_{+}$ are in $[0,1/2]$. The constraints of the special gates $G_{()^2}, G_{-}^{[0, 1]}, G_{*2}^{[0,1]}$ are shown in Table \ref{tbl:special_gates}.

In general, the input of $H$ is a $N$-dimensional vector $x \in [0,1]^N$ is given by $N$ nodes with in-degree 0 and out-degree 1, called \emph{input-nodes}. Also, in general, the output of $H$ is a $M$-dimensional vector $x' \in [0,1]^M$ (the dimension of the circuit's output is of no importance here). Moreover, it could be the case that $H$ is \emph{cyclic}, meaning that it has no input and no output, but here we will consider the general case.  Without loss of generality, let the rest of the nodes be of in-degree 1 and out-degree 1, located right after each gate's output. By ``right after'' we mean that if a gate's output has a branching, the node is placed before the branching.
Suppose that the total number of nodes in $H$ is $r := N + |\tcal| = poly(N)$, since by definition $H$ has polynomial size.

If the node $v_{i} \in V$ for $i \in [r]$ is at the output of gate $g_i$ we will call it the \textit{output-node of $g_i$} (otherwise it will be an input-node). For an example see Figure \ref{pl-node2}.
\begin{figure}
	\begin{center}
		\includegraphics[scale=0.80]{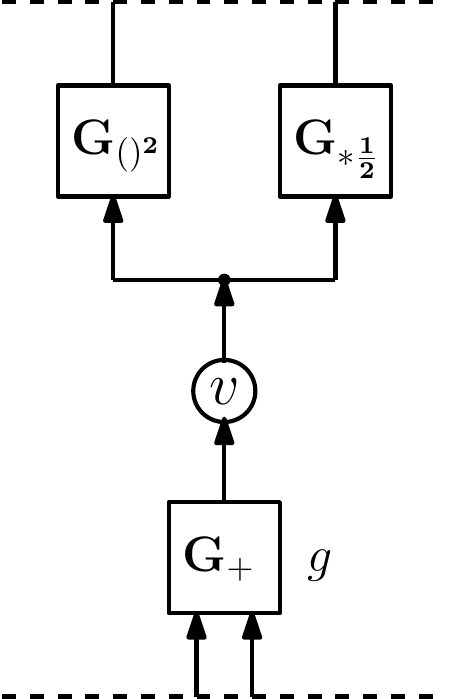}
	\end{center}
	\caption[A gadget to turn an arithmetic circuit into a \CH instance.]{A node in series with the output of an addition gate. $v$ is the \textit{output-node} of $g$.}\label{pl-node2}
\end{figure}

Consider the node $v_i$, the output-node of gate $g_i$. $v_i$ corresponds to 4 Consensus Halving agents, named $ad_i$, $mid_i$, $cen_i$ and $ex_i$. Player $ad_i$ (Latin for ``to'') represents the incoming edge \emph{to} node $v_i$ and agent $ex_i$ (Latin for ``from'') the outgoing edge \emph{from} $v_i$, while both $mid_i$ and $cen_i$ represent an edge at the \emph{middle} (\emph{center}) of node $v_i$ that connects its input and output. The number of agents created in $H$ is $n := 4r$. The domain of the valuation functions of the agents is $[0,12r]$. Furthermore, this interval is split to $r$ blocks, with the $i$-th block being $[b_i, b_{i+1}]$, where $b_i := 12(i-1)$, $i \in [r]$. 

According to the definition of the \CH problem, the domain of the valuation functions of the agents is $[0,1]$. Although the domain of the valuation functions of the \CH instance that we reduce to is $[0,12r]$, this is just for convenience of presentation. In fact, by scaling down each block to length $1/(12r)$ (divide by $12r$), the domain becomes $[0,1]$ and the correctness of the reduction is preserved. 

Let us define the function $border_{i}(t)$, $t \in [0,12r]$ for each node $v_i$, $i \in [r]$. The idea for this function is from \cite{FFGZ18}. If $v_i$ is the output-node of gate type $G_{*\zeta}$, then
\begin{align*}
	border_{i}(t) &= \begin{cases}
		4 , \quad &t \in [b_{i}, b_{i}+1] \cup [b_{i}+1+\zeta, b_{i}+2+\zeta] \\
		0 , \quad &\text{otherwise}
	\end{cases}
\end{align*}
If $v_i$ is the output-node of any gate type other than $G_{*\zeta}$, then
\begin{align*}
	border_{i}(t) &= \begin{cases}
		4 , \quad &t \in [b_{i}, b_{i}+1] \cup [b_{i}+2, b_{i}+3] \\
		0 , \quad &\text{otherwise}
	\end{cases}
\end{align*}
and also:
\begin{itemize}
	\item $v_{i}^{a} := [b_{i} + 1, b_{i} + 2] := [v_{i,l}^{a}, v_{i,r}^{a}]$
	\item $v_{i}^{m} := [b_{i} + 4, b_{i} + 5] := [v_{i,l}^{m}, v_{i,r}^{m}]$
	\item $v_{i}^{-} := [b_{i} + 7, b_{i} + 8] := [v_{i,l}^{-}, v_{i,r}^{-}]$
	\item $v_{i}^{+} := [b_{i} + 10, b_{i} + 11] := [v_{i,l}^{+}, v_{i,r}^{+}]$
	\item $G_{\pi}(t)$ is the function corresponding to gate of type $G_{\pi} \in \{ G_{\zeta}, G_{*\zeta}, G_{+}, G_{()^2}, G_{-}^{[0,1]}, G_{*2}^{[0,1]}\}$ (see Figure \ref{gates-table}).
\end{itemize}

The valuation functions of the agents $ad_i$, $mid_i$, $cen_i$ and $ex_i$ corresponding to node $v_i$ are,
\begin{align}\label{val-functions}
	ad_i(t) &= \begin{cases}
		border_{i}(t) + G_{\pi}(t) , \quad &\text{if $v_i$ is the output-node of gate type $G_{\pi}$}\\
		border_{i}(t) , \quad &\text{if $v_i$ is input-node (input of $H$)}.
	\end{cases} \\
	mid_i(t) &= \begin{cases}
		4 , \quad &t \in [b_{i}+3, b_{i}+4] \cup [b_{i}+5, b_{i}+6] \\
		1 , \quad &t \in v_{i}^{a} \cup v_{i}^{m} \\
		0 , \quad &\text{otherwise}
	\end{cases} \nonumber \\
	cen_i(t) &= \begin{cases}
		4 , \quad &t \in [b_{i}+6, b_{i}+7] \cup [b_{i}+8, b_{i}+9] \\
		1 , \quad &t \in v_{i}^{m} \cup v_{i}^{-} \\
		0 , \quad &\text{otherwise}
	\end{cases} \nonumber \\
	ex_i(t) &= \begin{cases}
		4 , \quad &t \in [b_{i}+9, b_{i}+10] \cup [b_{i}+11, b_{i}+12] \\
		1 , \quad &t \in v_{i}^{-} \cup v_{i}^{+} \\
		0 , \quad &\text{otherwise}
	\end{cases} \nonumber
\end{align}

The intuition for the synergy of the 4 agents is the following: Take as a given that in a solution of the created \CH instance with at most $n$ cuts, a cut is placed only (almost always\footnote{With the only exception being a cut before $v_{i}^{a}$ when gate $g_i$ is $G_{-}^{[0,1]}$ and its result is negative. See Figure \ref{example_minus} for an example.}) in the intervals $v_{i}^{a}, v_{i}^{m}, v_{i}^{-}, v_{i}^{+}$ for every $i \in [r]$. Since the length of each of those intervals is 1, each such cut encodes a number in $[0,1]$. Consider $v_i$, the output-node of gate $g_i$ with inputs $v_{j}, v_{k}$. Think of the agents $ad_i$, $mid_i$, $cen_i$, $ex_i$ as being sequential, meaning that each of them ``computes'' a value via a cut in $v_{i}^{a}, v_{i}^{m}, v_{i}^{-}$ or $v_{i}^{+}$ respectively, and feeds it in the next agent. In particular, agent $ad_i$ takes as input the values (in the form of cuts) that nodes $v_{j}, v_{k}$ give her, and computes the exact operation that $g_i$ prescribes (e.g. if $g_i$ is type $G_{-}^{[0,1]}$, $ad_i$ performs subtraction of the input values without capping at 0, see Figure \ref{example_minus}). Then $ad_i$ feeds this value in $mid_i$ via creating a cut in $v_{i}^{a}$, and $mid_i$ computes the actual value in $[0,1]$ that $g_i$ should output (e.g. if $g_i$ is type $G_{-}^{[0,1]}$, in this step $mid_i$ caps the value at 0), and feeds it in $cen_i$ via creating a cut in $v_{i}^{m}$. 
This correct value should be exported for further use from other gates to which $v_i$ is input, but depending on these gates, the positive or negative of that value might be needed (by ``positive'' and ``negative'' we mean the label, not the actual sign of the value). That is why a negative version of this value is produced by $cen_{i}$ and a positive by $ex_i$, via a cut in $v_{i}^{-}$ and $v_{i}^{+}$ respectively. A negative(resp. positive) value is one encoded by a cut that defines an interval at its left which is \emph{negative}(resp \emph{positive}). Moreover, for every input-node $v_j$ we arbitrarily consider $ad_j$ to encode a negative value, therefore, since (by the structure of the \CH instance) the labels of the values induced by the 4 agents are alternating, the agents $mid_{i}$, $cen_{i}$, $ex_{i}$ encode a positive, negative, and positive value, respectively.

\subsection[One-to-two correspondence of circuit values to CH cuts]{One-to-two correspondence of circuit values to \CH cuts}

Here we show that a solution of the special circuit maps to one pair of \CH solutions (since the solutions come by definition in pairs of opposite signs of pieces), and any pair of \CH solutions maps to exactly one solution of the special circuit.

Let us define the functions $z_{i}(x)$, $i \in [r]$ that depend on the input vector $x \in [0,1]^N$, and compute the value of each node $v_i$ of the arithmetic circuit $H$. Let us also, without loss of generality, set $(z_1, \dots, z_N) := (x_1, \dots, x_N)$. First, we will show that for every tuple $(z_{1}(x), \dots, z_{r}(x))$ of values that satisfy $H$, a solution in the constructed \CH instance with $n$ agents and $n$ cuts ($n := 4r$) encodes the same values via its cuts. We will then show that for every solution of the \CH instance with $n$ agents and $n$ cuts, the cuts correspond to a unique tuple $(z_1, \dots, z_r)$ that satisfies $H$. 

In the sequel, we call a cut $t$ \emph{negative}(resp. \emph{positive}) if the interval that it defines at its left has negative(resp. positive) label. Also, in the following subsections, the analysis is done for the case where the resulting \CH solution has its leftmost interval being negative. However, there is one more solution symmetric to this, in which the leftmost interval is positive. In any solution we remind that the intervals are of alternating signs (see definition of a \CH solution in Section \ref{sec:CH_definition}). We omit the analysis of the solution where the leftmost interval is positive since it is identical to the presented one.

\subsubsection{Circuit values to cuts}\label{par:cir-to-ch}

Suppose the tuple $(z_{1}^*, \dots, z_{r}^*)$ satisfies $H$. We will show that from this solution we can create a \CH solution with $n:=4r$ cuts, i.e. all of the agents are satisfied. Consider node $v_{i}$ of $H$. Let us translate the values $z_{i}^{*}$, $i \in [r]$ into cuts as follows:

\begin{itemize}
	\item If $g_{i}$'s type is one of $G_{\zeta}, G_{*\zeta}, G_{+}, G_{()^2}, G_{*2}^{[0,1]}$ or $v_i$ is an input-node. 
	\begin{itemize}
		\item Place a cut at $t = v_{i,l}^{a} + z_{i}^{*}$,
		\item Place a cut at $t = v_{i,l}^{m} + z_{i}^{*}$,
		\item Place a cut at $t = v_{i,l}^{-} + z_{i}^{*}$,
		\item Place a cut at $t = v_{i,l}^{+} + z_{i}^{*}$.
	\end{itemize}
	\item If $g_{i}$'s type is $G_{-}^{[0,1]}$, i.e. $g_{i} = \max\{g_j - g_k , 0\}$, and $z_j^* \geq z_k^*$. 
	\begin{itemize}
		\item Place a cut at $t = v_{i,l}^{a} + z_{i}^{*}$,
		\item Place a cut at $t = v_{i,l}^{m} + z_{i}^{*}$,
		\item Place a cut at $t = v_{i,l}^{-} + z_{i}^{*}$,
		\item Place a cut at $t = v_{i,l}^{+} + z_{i}^{*}$.
	\end{itemize}
	\item If $g_{i}$'s type is $G_{-}^{[0,1]}$, i.e. $g_{i} = \max\{g_j - g_k , 0\}$, and $z_j^* < z_k^*$
	\begin{itemize}
		\item Place a cut at $t = v_{i,l}^{a} - (z_{k}^{*} - z_{j}^{*})/5$,
		\item Place a cut at $t = v_{i,l}^{m} + z_{i}^{*}$,
		\item Place a cut at $t = v_{i,l}^{-} + z_{i}^{*}$,
		\item Place a cut at $t = v_{i,l}^{+} + z_{i}^{*}$.
	\end{itemize}
\end{itemize}

%
%Suppose the tuple $(z_{1}^*, \dots, z_{r}^*)$ satisfies $C$. We will show that from this solution we can create a \CH solution with $n:=4r$ cuts, i.e. all of the agents are satisfied. Consider node $v_{i}$ which is the output-node of $g_{i}$, where $g_{i}$'s type is one of $G_{\zeta}, G_{*\zeta}, G_{+}, G_{()^2}, G_{-}^{[0,\frac{1}{2}]}, G_{*2}^{[0,1]}$ (not $G_{-}^{[0,1]}$) or it is an input-node. Let us translate the values $z_{i}^{*}$, $i \in [r]$ into cuts as follows: 
%\begin{itemize}
%	\item Place a cut at $t = v_{i,l}^{a} + z_{i}^{*}$,
%	\item Place a cut at $t = v_{i,l}^{m} + z_{i}^{*}$,
%	\item Place a cut at $t = v_{i,l}^{-} + z_{i}^{*}$,
%	\item Place a cut at $t = v_{i,l}^{+} + z_{i}^{*}$.
%\end{itemize}
%	
By construction of the valuation functions of the agents, these cuts are placed one after the other, where there is one cut in each of the intervals $v_{i}^{a}$, $v_{i}^{m}$, $v_{i}^{-}$, $v_{i}^{+}$ in that order, and each such sequence of four cuts is in an increasing order of $i$.
By definition, any solution of \CH has alternating signs of the resulting pieces, and, as mentioned earlier, each solution comes with another symmetric solution with the same cuts and opposite signs of pieces. The analysis here is shown for the solution where the leftmost piece is negative, and we omit the analysis of the symmetric solution since it is identical.
%
% However, we remark that only the positions of the cuts encode the values of the circuit, and not the signs of the pieces. That is because of two basic properties of any solution of our constructed instance: (a) there are cuts only in intervals $v_{i}^{a}$, $v_{i}^{m}$, $v_{i}^{-}$, $v_{i}^{+}$ for $i \in [r]$ and there is exactly one cut in each, therefore each node that is the output of each gate $i$ is encoded by a group of 4 cuts in the aforementioned intervals in a \CH solution; (b) agent $i$  the sign of each piece is defined so as to satisfy the  (b) the placement of the cuts on each agent's positive parts of the valuation function is uniquely determined by the values of the input nodes of the circuit which enforce the values to the rest of the circuit's nodes according to the operations of the gates.  , and only the positions of the cuts encode these values, not of the circuit the valuation functions of $ad_i$, $mid_i$, $cen_i$, $ex_i$ for any $i \in [r]$ , if the cut placed in $v_{i}^{a}$ (or slightly outside $v_{i}^{a}$ in case $i$ is a $G_{-}^{[0,1]}$ gate) is negative then the only way for agent $ad_i$ to be satisfied is  We show only the analysis for the case where the solution has negative leftmost interval. We remark that the only solutions are this and another one symmetric to it with positive leftmost interval, whose cuts are at the exact same points; we omit this analysis since it is identical to the presented one. 

Let us now prove that for every $i \in [r]$, the $ad_i$ agent is satisfied. 

\paragraph*{$\mathbf{G_{\zeta}:}$}
This gate has no input. Consider its output $z_{i}^* = \zeta$ and its output-node $v_i$. By our constructed $n$-cut, a cut is placed at $t = v_{i,l}^{a} + \zeta$. Since the valuation function of $ad_i$ is symmetric around $v_{i, l}^{a} + \zeta$ (see aforementioned equation \eqref{val-functions} that describes the valuation functions), the total valuation is cut exactly in half (see Figure \ref{gates-table}), therefore agent $ad_i$ is satisfied. 
%Agents $mid_i$, $cen_i$, $ex_i$ are also satisfied since they just copy the value by placing cuts at $t = v_{i,l}^{m} + \zeta$, $t = v_{i,l}^{-} + \zeta$ and $t = v_{i,l}^{+} + \zeta$ respectively.

\paragraph*{$\mathbf{G_{*\zeta}:}$}
Consider its input $z_{j}^*$, output $z_{i}^* = \zeta \cdot z_{j}^*$ and its output-node $v_i$. By our constructed $n$-cut, a positive cut is placed at $t = v_{j,l}^{+} + z_{j}^*$ and a negative cut is placed at $t = v_{i,l}^{a} + z_{i}^*$. Agent $ad_i$ is satisfied since her positive valuation equals her negative one. In particular, $z_{j}^* \cdot 1 + (\zeta - z_{i}^{*}) \cdot \frac{1}{\zeta} + 1 \cdot 4 = (1 - z_{j}^*) \cdot 1 + 1 \cdot 4 + z_{i}^{*} \cdot \frac{1}{\zeta}$ is true.
%By construction of the valuation functions of agents $mid_i$, $cen_i$, $ex_i$ in order to be satisfied they just copy the value by placing cuts at $t = v_{i,l}^{m} + \zeta \cdot z_{j}^*$, $t = v_{i,l}^{-} + \zeta \cdot z_{j}^*$ and $t = v_{i,l}^{+} + \zeta \cdot z_{j}^*$ respectively.

\paragraph*{$\mathbf{G_{+}:}$}
Consider its inputs $z_{j}^*$, $z_{k}^*$, its output $z_{i}^* = z_{j}^* + z_{k}^*$ and its output-node $v_i$. By our constructed $n$-cut, a positive cut is placed at $t = v_{j,l}^{+} + z_{j}^*$, another positive cut is placed at $t = v_{k,l}^{+} + z_{k}^*$ and a negative cut is placed at $t = v_{i,l}^{a} + z_{i}^*$. Agent $ad_i$ is satisfied since her positive valuation equals her negative one. In particular, $z_{j}^* \cdot 1 + z_{k}^* \cdot 1 + (1 - z_{i}^*) \cdot 1 + 1 \cdot 4 = (1/2 - z_{j}^*) \cdot 1 + (1/2 - z_{k}^*) \cdot 1 + 1 \cdot 4 + z_{i}^{*} \cdot 1$ is true.
%By construction of the valuation functions of agents $mid_i$, $cen_i$, $ex_i$ in order to be satisfied they just copy the value by placing cuts at $t = v_{i,l}^{m} + z_{j}^* + z_{k}^*$, $t = v_{i,l}^{-} + z_{j}^* + z_{k}^*$ and $t = v_{i,l}^{+} + z_{j}^* + z_{k}^*$ respectively.

\paragraph*{$\mathbf{G_{()^2}:}$}
Consider its input $z_{j}^*$, output $z_{i}^* = (z_{j}^*)^2$ and its output-node $v_i$. By our constructed $n$-cut, a positive cut is placed at $t = v_{j,l}^{+} + z_{j}^*$ and a negative cut is placed at $t = v_{i,l}^{a} + z_{i}^*$. Agent $ad_i$ is satisfied since her positive valuation equals her negative one. In particular, $(z_{j}^*)^2 + (1 - z_{i}^{*}) \cdot 1 + 1 \cdot 4 = (1 - (z_{j}^*)^2) + 1 \cdot 4 + z_{i}^{*} \cdot 1$ is true.
%By construction of the valuation functions of agents $mid_i$, $cen_i$, $ex_i$ in order to be satisfied they just copy the value by placing cuts at $t = v_{i,l}^{m} + (z_{j}^*)^2$, $t = v_{i,l}^{-} + (z_{j}^*)^2$ and $t = v_{i,l}^{+} + (z_{j}^*)^2$ respectively.

%\paragraph*{$\mathbf{G_{-}^{[0,\frac{1}{2}]}:}$}
%Consider its inputs $z_{j}^*$, $z_{k}^*$, its output $z_{i}^* = z_{j}^* - z_{k}^*$ and its output-node $v_i$. According to the design of our circuit, any positive input to the special subtraction gate is at most 1 and any negative input to it is at most 1/2 (see Figure \ref{mult-gate}). By our constructed $n$-cut, a positive cut is placed at $t = v_{j,l}^{+} + z_{j}^*$, a negative cut is placed at $t = v_{k,l}^{-} + z_{k}^*$ and another negative cut is placed at $t = v_{i,l}^{a} + z_{i}^*$. The valuation function is symmetric around $v_{i}^{a}$ (see (\ref{val-functions})), therefore, in order for $ad_i$ to be satisfied, it suffices that $z_{j}^* \cdot 1 + (1/2 - z_{k}^*) \cdot 1 +  (1/2 - z_{i}^*) \cdot 1 = (1 - z_{j}^*) \cdot 1 + z_{k}^* \cdot 1 + z_{i}^{*} \cdot 1$, which is true.

\paragraph*{$\mathbf{G_{*2}^{[0,1]}:}$}
Consider its input $z_{j}^*$, output $z_{i}^* = 2 \cdot z_{j}^*$ and its output-node $v_i$. By our constructed $n$-cut, a positive cut is placed at $t = v_{j,l}^{+} + z_{j}^*$ and a negative cut is placed at $t = v_{i,l}^{a} + z_{i}^*$. Agent $ad_i$ is satisfied since her positive valuation equals her negative one. In particular, $z_{j}^* \cdot 1 + (1 - z_{i}^{*}) \cdot \frac{1}{2} + 1 \cdot 4 = (1/2 - z_{j}^*) \cdot 1 + 1 \cdot 4 + z_{i}^{*} \cdot \frac{1}{2}$ is true.

\paragraph*{$\mathbf{G_{-}^{[0,1]}:}$}
Consider its inputs $z_{j}^*$, $z_{k}^*$, its output $z_{i}^* = \max \{ z_{j}^* - z_{k}^* , 0 \}$ and its output-node $v_i$. By our constructed $n$-cut, 
\begin{itemize}
	\item if $z_j^* \geq z_k^*$, then $z_{i}^* = z_{j}^* - z_{k}^*$. By our constructed $n$-cut, a positive cut is placed at $t = v_{j,l}^{+} + z_{j}^*$, a negative cut is placed at $t = v_{k,l}^{-} + z_{k}^*$ and another negative cut is placed at $t = v_{i,l}^{a} + z_{i}^*$. Agent $ad_i$ is satisfied since her positive valuation equals her negative one. In particular, $z_{j}^* \cdot 1 + (1 - z_{k}^*) \cdot 1 +  (1 - z_{i}^*) \cdot 1 + 1 \cdot 4 = (1 - z_{j}^*) \cdot 1 + z_{k}^* \cdot 1  + 1 \cdot (1+4) + z_{i}^{*} \cdot 1$ is true.
	\item if $z_j^* < z_k^*$, then $z_{i}^* = 0$. By our constructed $n$-cut, a positive cut is placed at $t = v_{j,l}^{+} + z_{j}^*$, a negative cut is placed at $t = v_{k,l}^{-} + z_{k}^*$ and another negative cut is placed at $t = v_{i,l}^{a}   - (z_{k}^{*} - z_{j}^{*})/5$. Agent $ad_i$ is satisfied since her positive valuation equals her negative one. In particular, $z_{j}^* \cdot 1 + (1 - z_{k}^*) \cdot 1 + \frac{z_{k}^* - z_{j}^*}{5} \cdot (1+4) + 1 \cdot 1 + 1 \cdot 4  = (1 - z_{j}^*) \cdot 1 + z_{k}^* \cdot 1 + (1 - \frac{z_{k}^* - z_{j}^*}{5}) \cdot (1+4)$ is true.
\end{itemize} 
~\\

We will now prove that in our constructed $n$-cut, the agents $mid_i$, $cen_i$, $ex_i$ are also satisfied. If $g_i$ is not a $G_{-}^{[0,1]}$ gate, let us prove that $mid_i$ is satisfied. In our $n$-cut there is a negative cut at $t = v_{i,l}^{a} + z_{i}^{*}$ and a positive one at $t = v_{i,l}^{m} + z_{i}^{*}$. Agent $mid_i$ is satisfied since her negative valuation equals her positive one. In particular, $z_{i}^* \cdot 1 + (1 - z_{i}^*) \cdot 1 + 1 \cdot 4  = (1 - z_{i}^*) \cdot 1 + 1 \cdot 4 + z_{i}^* \cdot 1 $ is true.

If $g_i$ is a $G_{-}^{[0,1]}$ gate, let us prove that $mid_i$ is satisfied. 
\begin{itemize}
	\item if $z_j^* \geq z_k^*$, then a negative cut is placed at $t = v_{i,l}^{a} + z_{i}^*$, and a positive cut is placed at $t = v_{i,l}^{m} + z_{i}^*$. Agent $mid_i$ is satisfied since her negative valuation equals her positive one. In particular, $z_{i}^* \cdot 1 + (1 - z_{i}^*) \cdot 1 + 1 \cdot 4  = (1 - z_{i}^*) \cdot 1 + 1 \cdot 4 + z_{i}^* \cdot 1 $ is true.
	\item if $z_j^* < z_k^*$, then a negative cut is placed at $t = v_{i,l}^{a} - (z_{k}^{*} - z_{j}^{*})/5$ and a positive cut is placed at $t = v_{i,l}^{m}$. Agent $mid_i$ is satisfied since her negative valuation equals her positive one. In particular, $ \frac{z_{k}^{*} - z_{j}^{*}}{5} \cdot 0 + 1 \cdot 1 + 1 \cdot 4 = 1 \cdot 1 + 1 \cdot 4 + 0 \cdot 1 $ is true.
\end{itemize} 
For the agents $cen_i$ and $ex_i$, since their valuation functions are the same as $mid_i$ shifted to the right, it is easy to see that the $n$-cut we provide forces them to have positive valuation equal to their negative one.

\subsubsection{Cuts to circuit values}\label{par:ch-to-cir}

Now suppose that the tuple $(t_{1}^*, \dots, t_{n}^*)$ with $0 \leq t_{1}^* \leq \dots \leq t_{n}^* \leq 12r$, represents an $n$-cut ($n := 4r$) that is a solution of the constructed \CH instance with $n$ agents, where w.l.o.g. the first $4N$ cuts correspond to the $N$ input-nodes. We will show that from this solution we can construct a tuple $(z_{1}, \dots, z_{r})$ that satisfies the circuit $H$. Again, note that solutions of \CH come in pairs, where the two solutions have the same cuts but opposite signs of pieces. We will only analyse the solution where the leftmost interval has negative sign and omit the symmetric case with positive leftmost interval since the analysis is identical and both \CH solutions map to the same circuit solution.

Consider node $v_{i}$ which is the output-node of gate $g_{i}$ or it is an input-node. Observe that the valuation function of each of $ad_i, mid_i, cen_i$ and $ex_i$ has more than half of her total valuation inside the interval $[b_i , b_i + 3], [b_i + 3, b_i + 6], [b_i + 6, b_i + 9]$ and $[b_i + 9, b_i + 12]$ respectively. This means that in a solution, each of them has to have at least one cut in her corresponding aforementioned interval. But since these intervals are not overlapping for all $n$ agents, and we need to have at most $n$ cuts, exactly one cut has to be placed by each agent in her corresponding interval. 

Consider now the first $4N$ cuts that correspond to the input-nodes. As it is apparent from the definition of these nodes' valuation functions, each agent of $ad_i, mid_i, cen_i, ex_i$ for $i \in [N]$ has to place her single cut in the interval $v_{i}^{a}, v_{i}^{m}, v_{i}^{-}, v_{i}^{+}$ respectively. Given the latter fact, the definition of valuation functions for non input-node agents dictates that there will always be a cut in $v_{i}^{+}$ for every $i \in [r]$. Since $0 \leq t_{1}^* \leq \dots \leq t_{n}^* \leq 12r$, the sequential nature of our agents indicates that the cut $t_{4i}^*$, i.e. with index $4 \cdot i$, is found in interval $v_{i}^{+}$. Now, let us translate the position of the cut $t_{4i}^*, i \in [r]$ into the value $z_i = t_{4i}^* - v_{i,l}^{+} $. By a similar argument as that of the previous paragraph showing that the $ad_i$ agents are satisfied, it is easy to see that, by the aforementioned translation, the created tuple $(z_{1}, \dots, z_{r})$ satisfies circuit $H$.

\subsubsection{Valuation functions to circuits}\label{par:fun-to-cir}

In the \CH instances we construct, we have described the valuation functions of the agents mathematically. However, in a \CH instance the input is an arithmetic circuit, therefore we have to turn each valuation function of each agent $j \in [n]$ into its integral, and subsequently into an arithmetic circuit. Here we describe a method to do that.

The valuation functions we construct in our reduction (see Section \ref{app:cir_to_ch}) are piecewise polynomial functions of a single variable and their degree is at most 1, with $k$ pieces where $k$ is constant. Therefore, their integrals, which are the input of the \CH problem (captured by arithmetic circuits), are piecewise polynomial functions (with the same pieces) with degree at most 2. 
Consider the valuation function $f$ of an arbitrary player. Let the pieces of $f$ be $[p_0,p_1), [p_1,p_2), \dots , [p_{k-1},p_{k}]$ where $p_0 = 0$ and $p_k = 1$ and denote $P_1, P_2, \dots, P_k$ the above pieces respectively. Let us also denote by $f^{P_s}$ the polynomial in interval $P_s$, $s \in \{1,2, \dots, k\}$. In particular, $f$ can be defined as
\begin{align}\label{f(t)}
	f(t) = \begin{cases}
		f^{P_1}(t) \quad, t \in [p_0,p_1) \\
		f^{P_2}(t) \quad, t \in [p_1,p_2) \\
		\vdots \\
		f^{P_k}(t) \quad, t \in [p_{k-1},p_{k}],
	\end{cases}
\end{align}
and according to the valuation functions used in the reduction (see Section \ref{app:cir_to_ch}), for any given piece $P_s$ there are two kinds of possible functions

\begin{enumerate}
	\item[(a)] $f^{P_s}(t) = c_s$, where $c_s \geq 0$ is a constant, or
	\item[(b)] $f^{P_s}(t) = 2 \cdot ( t - p_{s-1} )$.
\end{enumerate}
(The latter comes from the valuation function of an $ad$ agent that corresponds to an output node of a $G_{()^2}$ gate.)

We would like to find a formula for the integral of $f(t)$, denoted $F(t)$, and we also require that $F(t)$ is computable by an arithmetic circuit, so that it is a proper input (together with the other agents' integrals of valuation functions) to the \CH instance.
For each piece $P_s$ we will construct an integral, denoted by $F^{P_s}(t)$, such that each such integral will be computable by an arithmetic circuit, and so that it will be $F(t) = \sum_{s \in \{1,2,\dots,k\}} F^{P_s}(t)$.
First, let us construct the function $D_{s}(t)$ using the domain $P_s$ of $f^{P_s}(t)$:
\begin{align*}
	D_{s}(t) := \min\left\{ \max\left\{ t, p_{s-1} \right\}, p_{s} \right\},
\end{align*}
which takes values
\begin{align*}
	D_{s}(t) = \begin{cases}
		p_{s-1} , &t < p_{s-1} \\
		t , &t \in [p_{s-1},p_{s}] \\
		p_{s} , &t > p_{s}.
	\end{cases}
\end{align*}

Now, for function $f^{P_s}(t)$ of case (a), we construct its integral:
\begin{align*}
	F^{P_s}(t) := c_{s} \cdot \left( D_{s}(t) - p_{s-1} \right),
\end{align*}
which takes values
\begin{align*}
	F^{P_s}(t) = \begin{cases}
		0 , &t < p_{s-1} \\
		c_{s} \cdot \left( t - p_{s-1} \right) , &t \in [p_{s-1},p_{s}] \\
		c_{s} \cdot \left( p_{s} - p_{s-1} \right) , &t > p_{s}.
	\end{cases}
\end{align*}
Similarly, for function $f^{P_s}(t)$ of case (b), we also construct its integral:
\begin{align*}
	F^{P_s}(t) := \left( D_{s}(t) - p_{s-1} \right)^2,
\end{align*}
which takes values
\begin{align*}
	F^{P_s}(t) = \begin{cases}
		0 , &t < p_{s-1} \\
		\left( t - p_{s-1} \right)^2 , &t \in [p_{s-1},p_{s}] \\
		\left( p_{s} - p_{s-1} \right)^2 , &t > p_{s}.
	\end{cases}
\end{align*}

Finally, for the agent with valuation function $f(t)$, the corresponding function computable by the arithmetic circuit that is input to the \CH problem is:
\begin{align*}
	F(t) := \sum_{s \in \{1,2,\dots,k\}} F^{P_s}(t).
\end{align*}
For the integral function $F(t)$ indeed it holds that $F(t) = \int_0^t f(x) \, dx$ as required. That is because, by the way we defined each $F^{P_s}(t)$, for any $t \in P_{s^*}$ it is
\begin{align*}
	F(t) = \sum_{s \in \{1,2,\dots,k\}} F^{P_s}(t) &= \sum_{s \in \{1,2,\dots,s^{*}-1\}} F^{P_s}(t) + F^{P_s^{*}}(t) + \sum_{s \in \{s^{*}+1,\dots,k\}} 0 \\
	&= \sum_{s \in \{1,2,\dots,s^{*}-1\}} \int_{P_s} f^{P_s}(x) \, dx +  \int_{p_{s^{*}-1}}^{t} f^{P_s^{*}}(x) \, dx \\
	&= \sum_{s \in \{1,2,\dots,s^{*}-1\}} \int_{p_{s-1}}^{p_{s}} f(x) \, dx +  \int_{p_{s^{*}-1}}^{t} f(x) \, dx \\
	&= \int_0^t f(x) \, dx
\end{align*}

For each player with some valuation function $f$ as defined above, we can compute the functions $F^{P_s}$, $s \in [k]$ by using gates $G_{\zeta}, G_{*\zeta}, G_{-}, G_{*}, G_{min}, G_{max}$. Then $F(t)$ can be computed by using $G_{+}$ gates. The arithmetic circuits that compute the functions $F(t)$ (one for each agent $j \in [n]$) constitute a proper \CH instance.
This completes the proof of Lemma \ref{lem:circuit-embed}.

\section{Proof of Theorem \ref{thm:CH_is_FIXP-hard}}\label{app:CH_is_FIXP-hard}

In this section we give a detailed proof of Theorem \ref{thm:CH_is_FIXP-hard} which states that $(n,n)$-\CH is \FIXP-hard. This is accomplished by finding a polynomial-time reduction from the \FIXP-complete problem of computing a ``$d$-player Nash equilibrium'' to the $(n,n)$-\CH problem. Using the machinery of \cite{EY10}, the \FIXP-complete problem is first expressed as a circuit with particular properties, which is then embedded into a $(n,n)$-\CH instance using Lemma \ref{lem:circuit-embed}. We prove that all of these steps can be executed in polynomial time.

%In Theorem 4.7 of \cite{EY10} it is shown in great detail that the class \FIXP stays the same if the circuits (that capture the input of its problems) are restricted to using only the operations $\{ +, *, \max \}$ and rational constants. The ``$-$'' and ``$\min$'' operations are eliminated since multiplication with $-1$ is allowed. 

In particular, in \cite{EY10} it is shown that the problem of finding a Nash equilibrium of a $d$-player normal form game with $d \geq 3$ (``$d$-player Nash equilibrium'' problem) is \FIXP-complete. Given an instance of this problem, we will construct a polynomial-time reduction to $(n,n)$-\textsc{Consensus Halving}.
We will start from an arbitrary instance of ``$d$-player Nash equilibrium'' and, according to it, design a circuit using only the gates $G_{\zeta}, G_{+}, G_{-}, G_{*}, G_{\max}, G_{\min}$ with $\zeta \in \mathbb{Q}$. This step is done by a straightforward application of the procedure described in the proofs of Lemma 4.5 and Lemma 4.6 in \cite{EY10}. This circuit computes a function whose fixed points correspond precisely to the Nash equilibria of the initial game. Then, we create an equivalent circuit by ``breaking down'' the initial gates to some more suitable ones (by introducing ``special gates'', see Table \ref{tbl:special_gates}), whose inputs and outputs are guaranteed to be in $[0,1]$. From this, we will create a cyclic circuit, introduce Consensus Halving players on the ``wires'' of the circuit, and show that a Consensus Halving solution with at most as many cuts as the number of players in this instance can be efficiently translated back to a Nash equilibrium of the initial game.

\subsection{Expressing the game as a circuit without division gates}
Here, given an arbitrary $d$-player game, we will employ a function presented in \cite{EY10} whose fixed points are precisely the Nash equilibria of that game.
Consider a given instance $I$ of the ``$d$-player Nash equilibrium'' problem, i.e. a $d$-player normal form game where each player $i$ has a set $S_i$ of pure strategies. We will use the following notation similar to the one in \cite{EY10}: $N_i := |S_i|$, $N := \sum_{i}^{d} N_{i}$ and $v_{i}$ is the payoff function of player $i$ with domain $D_{I} := \times_{i=1}^{d} \Delta_{N_{i}}$, where $\Delta_{N_{i}}$ is the unit $(N_{i}-1)$-simplex. Define the \textit{mixed strategy profile} $x:=(x_{11}, \dots, x_{1N_1}, x_{21}, \dots, x_{2N_2}, \dots, x_{d1}, \dots, x_{dN_d})$ to be a $N$-dimensional vector with the entry $x_{ij}$ being the probability that player $i \in [d]$ plays pure strategy $j \in S_i$. Also, $v(x)$ is an $N$-dimensional vector with entries indexed as in $x$, with $v_{ij}(x) := v_{i}(j, x_{-i})$, the latter being the expected payoff of player $i$ when she plays the pure strategy $j \in S_i$ against the partial profile $x_{-i}$ of the rest of the players. The payoff function of each player is normalized by scaling in $[0,1/N]$ so that the Nash equilibria of the game are precisely the same. Thus, $v_{ij}(x) \in [0,1/N]$. Finally, let $h(x) := x + v(x)$.

Now, define for each player $i$ the function $f_{i,x}(t) := \sum_{j\in S_i} \max(h_{ij}(x) - t, 0)$ with parameter $x$. This function is defined in $\mathbb{R}$ and it is continuous, piecewise linear, strictly decreasing with values from $0$ to $+\infty$, thus there is a unique value $t_i \in \mathbb{R}$ such that $f_{i,x}(t_i)=1$. The required function whose set of fixed points is identical to the set of Nash equilibria of instance $I$ is $G_{I}(x)_{ij} := \max(h_{ij}(x)-t_i, 0)$ for $i \in [d]$, $j \in S_i$. The function $G_{I}$ takes as input the $n$-dimensional vector $x$ and outputs an $N$-dimensional vector $G_{I}(x)$ with entries defined as above. By definition of $G_{I}$ and choice of $t_i$, it is $\sum_{j \in S_i} G_{I}(x)_{ij} = 1$ for every $i \in [d]$, and therefore $G_{I}$ is a mapping of the domain $D_I$ to itself.

\begin{lemma}[LEMMA 4.5, \cite{EY10}]\label{lem: EY1}
	The fixed points of the function $G_{I}$ are precisely the Nash equilibria of the game $I$.
\end{lemma}
In fact, the structure of function $G_{I}$ allows for it to be efficiently constructed using only the required types of gates. 
\begin{lemma}[LEMMA 4.6, \cite{EY10}]
	We can construct in polynomial time a circuit with basis $\{ +, -, *, \max, \min \}$ (no division) and rational constants that computes the function $G_{I}$.
\end{lemma}
For the proofs of the above lemmata the reader is referred to the indicated work by Etessami and Yannakakis. 

In the proof of the latter lemma in \cite{EY10} it is shown how to construct an arithmetic circuit $C_{I}$ that computes the function $G_{I}$ using only gates of type $G_{\zeta}, G_{+}, G_{-}, G_{*}, G_{\max}, G_{\min}$, where $\zeta \in \mathbb{Q}$. The construction of $C_{I}$ is the following: Compute the function $y=h(x)=x + v(x)$ using only $G_{+}, G_{*}$ type of gates, allowed by the definition of $v(x)$. Vector $y$ has $d$ sub-vectors, where $y_{i}=(y_{i1}, y_{i2}, \dots, y_{iN_i})$. Then, each $y_i$ is sorted using a sorting network $Z_i$ thus creating a vector $z_{i}=(z_{i1}, z_{i2}, \dots, z_{iN_i})$ with sorted entries $z_{i1} \geq z_{i2} \geq \dots \geq z_{iN_i}$; sorting networks can be implemented in arithmetic circuits using only gates $G_{\max}, G_{\min}$ (for more see e.g. \cite{K98}). Using $z_{ij}$'s the function $t_{i}:= \max_{l \in [N_i]} \left\{ (1/l) * \left( \left( \sum_{j=1}^{l} z_{ij} \right) - 1 \right) \right\}$ is computed and the final output of the whole circuit is 
\begin{align}\label{output_C_{I}}
	x_{ij}' := \max \{ y_{ij} - t_i , 0 \} \quad \text{ for each } i \in [d], j \in S_{i}.
\end{align}

\subsection{A circuit with gates whose inputs/outputs are in $[0,1]$}
One can easily observe that some of the gates of circuit $C_{I}$ may have inputs and outputs outside of $[0,1]$. For example, the $G_{+}$ gate that computes $y_{ij}=x_{ij} + v(x)_{ij}$ can be 2 and the arguments of $G_{\max}$ in $t_i$ can be negative. We will transform this circuit into an equivalent one that guarantees its gates' inputs and outputs to be in $[0,1]$, using only gates $G_{\zeta}, G_{+}, G_{-}^{[0,1]}, G_{*}, G_{*2}^{[0,1]}, G_{\max}, G_{\min}$, where $\zeta \in \mathbb{Q}\cap(0,1]$.

In particular, instead of constructing the circuit $C_{I}$ as described in the previous paragraph, we will construct an equivalent one, called $C_{I}'$, whose input and output are the same as that of $C_{I}$, namely $x_{ij}$ and $x_{ij}'$, $i \in [d]$, $j \in [N_i]$ respectively, but its gates have inputs/outputs in $[0,1]$. We do this by manipulating the formula for the required function $G_I$ under computation, by suitably scaling up or down the input values of each gate, using additional gates $G_{\zeta}, G_{+}, G_{-}^{[0,1]} , G_{*}$. 

We construct $C_{I}'$ as follows: First, we compute the vector $p := h(x)/2 = x*\frac{1}{2} + v(x)*\frac{1}{2}$ using only $G_{+}, G_{*}$ gates. Note that $x_{ij}, v_{ij}(x), p_{ij} \in [0,1]$, $\forall i\in[d],j \in S_i$ (recall that the payoff function is normalized in $[0,1]$). Then, we sort each of the sub-vectors $p_{i}$, $i \in [d]$ via a sorting network $Q_i$ that can be constructed using $G_{\max}$ and $G_{\min}$ gates, thus computing the sorted vectors $q_{i} = (q_{i1}, q_{i2}, \dots, q_{iN_i})$ with sorted entries $q_{i1} \geq q_{i2} \geq \dots \geq q_{iN_i}$. Now, for every $i \in [d]$ and $l \in [N_i]$ we compute the following sub-function 
\begin{align*}
	t_{il}'' :=   \frac{1}{2} * \frac{1}{l} * \sum_{j=1}^{l} q_{ij}  + \frac{1}{2} - \frac{1}{4} * \frac{1}{l},
\end{align*} 
by using $l+1$ $G_{+}$ gates, 3 $G_{+}$ gates and 1 $G_{-}^{[0,1]}$ gate, where the subtraction gate is the last to take place. One should observe that since $\sum_{j=1}^{N_i} x_{ij} = 1$ and $\sum_{j=1}^{N_i} v_{ij}(x) \leq 1$ (by definition of payoff function in $[0,1/N]$), it is $\sum_{j=1}^{N_i} q_{ij}  \leq \frac{1}{2} \cdot \left( 1 + 1 \right) = 1$, therefore none of the individual computations of $t_{il}''$ is outside $[0,1]$.
% and also none of the addends is outside $[0,1/2]$. 
Moreover, in the subtraction, the value of the subtrahend is at most the value of the minuend so the subtraction is precise (not capped at 0).

Now, for each $i \in [d]$ we compute the sub-function 
\begin{align*}
	t_{i}'' := \max_{l \in [N_i]} \{ t_{il}'' \},
\end{align*} 
by using $N_i - 1$ $G_{\max}$ gates, and consequently compute
\begin{align*}
	t_{i}' := \left( t_{i}'' - \frac{1}{2} \right) * 2,
\end{align*} 
by using one $G_{-}^{[0,1]}$ and one special $G_{*2}^{[0,1]}$ gate where the computations happen from left to right. Note that $t_{i}'' \geq 1/2$, therefore the subtraction is precise (not capped at 0). Also, note that, by definition of $t_{il}''$, it is $t_{i}'' \leq 1$, therefore $t_{i}'' - 1/2 \leq 1/2$ and the output of the $G_{*2}^{[0,1]}$ gate of $t_{i}'$ is in $[0,1]$. Finally, the output of the circuit $C_{I}'$ is computed by
\begin{align}\label{output_C_{I}'}
	x_{ij}' := \max\{ p_{ij} - t_i', 0 \} * 2, \quad \text{ for each } i \in [d], j \in S_i,
\end{align} 
using one $G_{-}^{[0,1]}$ and one special $G_{*2}^{[0,1]}$ gate.

\begin{lemma}\label{lem:C'=C}
	Circuit $C_{I}'$ is equivalent to $C_{I}$, i.e. it computes the function $G_{I}$.
\end{lemma}

\begin{proof}
	We will show that for every $i \in [d], j \in S_i$, the value $x_{ij}$ of (\ref{output_C_{I}'}) is the same as that of (\ref{output_C_{I}}), i.e. the output of the circuits $C_{I}'$ and $C_{I}$ is the exact same. Using the formulas for $t_{il}'', t_{i}''$ and $t_{i}'$, we can re-write algebraically $x_{ij}$ by substituting the circuit's operations with the regular mathematical ones, i.e. $G_{+}, G_{-}^{[0,1]}, G_{*2}^{[0,1]}, G_{*}, G_{\max}, G_{\min}$ translate to $+, -, \cdot 2, \cdot, \max, \min$ respectively. Observe that this is possible since the $G_{-}^{[0,1]}$ gate, excluding the one in (\ref{output_C_{I}'}), actually performs subtraction without capping the output to 0. Thus, starting from (\ref{output_C_{I}'}) we have
	\begin{align*}
		x_{ij}' &= \max\{ p_{ij} - t_i', 0 \} \cdot 2 \\
		&= \max\{2 \cdot p_{ij} - 2 \cdot t_i', 0 \} \\
		&= \max\left\{y_{ij} - 4 \cdot \left( t_{i}'' - \frac{1}{2} \right), 0 \right\}   \quad \text{($y_{ij}$ from construction of $C_{I}$)} \\
		&= \max\left\{y_{ij} - 4 \cdot \left( \max_{l \in [N_i]} \{ t_{il}'' \} - \frac{1}{2} \right), 0 \right\}  \\
		&= \max\left\{y_{ij} - 4 \cdot \left( \max_{l \in [N_i]} \left\{ \frac{1}{2l} \cdot \left( \sum_{j=1}^{l} q_{ij} \right) + \frac{1}{2} - \frac{1}{4l}  \right\} - \frac{1}{2} \right), 0 \right\}  \\
		&= \max\left\{y_{ij} - 4 \cdot  \max_{l \in [N_i]} \left\{ \frac{1}{2l} \cdot \left( \sum_{j=1}^{l} q_{ij} \right) - \frac{1}{4l}  \right\} , 0 \right\}  \\
		&= \max\left\{y_{ij} -  \max_{l \in [N_i]} \left\{ \frac{1}{l} \cdot \left( \sum_{j=1}^{l} 2 \cdot q_{ij} \right) - \frac{1}{l}  \right\} , 0 \right\}  \\
		&= \max\left\{y_{ij} -  \max_{l \in [N_i]} \left\{ \frac{1}{l} \cdot \left( \left( \sum_{j=1}^{l} z_{ij} \right) - 1 \right) \right\} , 0 \right\}   \quad \text{($z_{ij}$ from construction of $C_{I}$)} \\
		&= \max\left\{y_{ij} -  t_i , 0 \right\}   \quad \text{($t_i$ from construction of $C_{I}$)},
	\end{align*} 
	which is by definition equal to the output $x_{ij}'$ of (\ref{output_C_{I}}).     
\end{proof}

The circuit $C_{I}'$ we constructed that computes the function $G_{I}$ uses gates of type in the set $\{G_{\zeta}, G_{+}, G_{*}, G_{\max}, G_{\min}, G_{-}^{[0,1]}, G_{*2}^{[0,1]} \}$, where $\zeta \in \mathbb{Q}\cap(0,1]$.

%
%
%\paragraph{\textbf{The final circuit.}}
%To make the reduction from the instance $I$ of the normal form game to a consensus halving instance, let us replace the $G_{\max}$ and $G_{\min}$ gates of our circuit, with $G_{\zeta}, G_{+}, G_{-}^{[0,1]}$ and $G_{*}$ gates, thus constructing a new circuit $C_{I}''$. In particular, for each gate $g_{i} := \max\{g_{j}, g_{k}\}$ that outputs the maximum of two inputs $g_{j}, g_{k} \in [0,1]$ we compute:
%%
%\begin{align*}
%\frac{1}{2} * g_{j} + \frac{1}{2} * g_{k} + \frac{1}{2} * \max\{g_j - g_k, 0\} + \frac{1}{2} * \max\{g_k - g_j, 0\},
%\end{align*} 
%%
%by using 4 $G_{\zeta}$, 4 $G_{*}$, 2 $G_{-}^{[0,1]}$ and 3 $G_{+}$ gates.
%
%Similarly, for each gate $g_{i} := \min\{g_{j}, g_{k}\}$ that outputs the minimum of two inputs $g_{j}, g_{k} \in [0,1]$ we compute:
%%
%\begin{align*}
%\frac{1}{2} * g_{j} + \frac{1}{2} * g_{k} - \frac{1}{2} * \max\{g_j - g_k, 0\} - \frac{1}{2} * \max\{g_k - g_j, 0\},
%\end{align*} 
%%
%by using 4 $G_{\zeta}$, 4 $G_{*}$, 4 $G_{-}^{[0,1]}$ and 1 $G_{+}$ gate. Note that if we perform the computations from left to right there is no way that the input/output of some gate will be outside $[0,1]$. 
%
%
%%%%%%%%%%%%%%%
%

\subsection{The $(n,n)$-\textsc{Consensus Halving} instance}
At this point we are ready to construct the $(n,n)$-\textsc{Consensus Halving} instance. The final circuit $C_{I}'$ computes the function $G_{I}$, where $G_{I}: D_{I} \to D_{I}$, whose fixed points are precisely the Nash equilibria of the initial instance $I$ of the $d$-player game, due to Lemma \ref{lem: EY1}. The output of $C_{I}'$ is the $N$-dimensional vector $x'$ with entries $x_{ij}'$ computed from (\ref{output_C_{I}'}). Let us close the circuit by connecting the output $x_{ij}'$ with the input $x_{ij}$ for every $i \in [d]$, $j \in S_i$. This new circuit, called $C_{I}^{o}$, is \emph{cyclic}, meaning that it has no input and no output.

The cyclic circuit $C_{I}^{o}$ (like $C_{I}'$) uses only gates in $\{G_{\zeta}, G_{+}, G_{*}, G_{\max}, G_{\min}, G_{-}^{[0,1]}, G_{*2}^{[0,1]} \}$, where $\zeta \in \mathbb{Q}\cap(0,1]$. In Section \ref{subsec: cir-to-ch} we describe how to turn such circuits into \textsc{Consensus Halving} instances. Suppose that $C_{I}^{o}$ uses $l$ gates. Then, by the procedure of Section \ref{subsec: cir-to-ch} let us turn $C_{I}^{o}$ into a special circuit $C_{I}^{o'}$ with $r = linear(l)$ gates which uses only the required gates by Lemma \ref{lem:circuit-embed}. Finally, still following that procedure, let us turn $C_{I}^{o'}$ into a \CH instance with $n := 4r$ agents. 

%
%The cyclic circuit $C_{I}^{o}$ (like $C_{I}'$) uses only gates in $\{G_{\zeta}, G_{+}, G_{*}, G_{\max}, G_{\min}, G_{-}^{[0,1]}, G_{*2}^{[0,1]} \}$, where $\zeta \in \mathbb{Q}\cap[0,1]$. In Section \ref{subsec: cir-to-ch} we describe how to turn such circuits into \CH instances. Suppose that $C_{I}^{o}$ uses $l$ gates. Then, by the procedure of Section \ref{subsec: cir-to-ch} let us turn $C_{I}^{o}$ into a special circuit $C_{I}^{o'}$ (resp. $H'$ in Section \ref{subsec:cir-to-ch}) with $r = linear(l)$ gates. Finally, still following that procedure, let us turn $C_{I}^{o'}$ into a \CH instance with $n := 4r$ agents. 

We can now prove Theorem~\ref{thm:CH_is_FIXP-hard}.

%\begin{theorem}
%$(n,n)$-\textsc{Consensus Halving} is \FIXP-hard. 
%\end{theorem}

\begin{proof} 
	%	Note that $C_{I}^{o}$ was created by connecting the output with the input of $C_{I}'$, i.e. by merging its output nodes $v'_1, v'_2, \dots, v'_N$ with its input nodes $v_1, v_2, \dots, v_N$ into a single tuple of nodes $(V_1, V_2, \dots, V_N)$. Imagine now that we do not merge the output and input nodes. Then note that the procedure in Section \ref{subsec:cir-to-ch} which turns $C_{I}^{o}$ into $C_{I}^{o'}$ preserves the value of $v'_1, v'_2, \dots, v'_N$, meaning that both circuits (if not made cyclic) compute the same output. Therefore, the value of the tuple $(v'_1, v'_2, \dots, v'_N)$ is the same as that of $C_{I}'$
	%	
	%	As proven in Section \ref{subsec:cir-to-ch}, a solution to that $(n,n)$-\CH instance, i.e. a solution with $n$ cuts, in linear time can be translated back to a tuple $z^* := (z_1^*, z_2^*, \dots, z_r^*)$ of satisfying values for the nodes of $C_{I}^{o'}$.
	%	
	%	 As we argued above, if $C_{I}^{o}$ and $C_{I}^{o'}$ were left acyclic, then their outputs would be the same. So, $x^*$ satisfies $C_{I}^{o}$
	%	
	%	
	%	 at the point of connection of the input and output of  the same  is an equivalent let us denote by $x'^* := (x_1'^*, x_2'^*, \dots, x_N'^*)$ the $N$ entries of $z^*$ that correspond to the output of $C_{I}'$.
	%	
	%
	%%	
	%	
	In Section \ref{app:circuit-embed} it was proven that a solution to the above $(n,n)$-\CH instance, i.e. a solution with $n$ cuts, in linear time can be translated back to a tuple $z^* := (z_1^*, z_2^*, \dots, z_r^*)$ of satisfying values for the nodes of $C_{I}^{o'}$. Recall that $C_{I}^{o'}$ was created by another cyclic equivalent circuit $C_{I}^{o}$ which was also created by merging the input and output nodes of an acyclic circuit $C_{I}'$. 
	
	Let us denote by $v_1, v_2, \dots, v_N$ and $v'_1, v'_2, \dots, v'_N$ the input and output nodes respectively of $C_{I}'$ and denote by $V_1, V_2, \dots, V_N$ the merged nodes in $C_{I}^{o}$ and $C_{I}^{o'}$. Let us denote by $x^* := (x_1^*, x_2^*, \dots, x_N^*)$ the $N$ entries of $z^*$ that correspond to the values of nodes $(V_1, V_2, \dots, V_N)$. Since the procedure in Section \ref{subsec: cir-to-ch} which turns $C_{I}^{o}$ into $C_{I}^{o'}$ preserves the computation of the values of $V_1, V_2, \dots, V_N$, it follows that $x^*$ satisfies $C_{I}^{o}$. Consequently, if the values $x^*$ are copied as values of both input $(v_1, v_2, \dots, v_N)$ and output $(v'_1, v'_2, \dots, v'_N)$ nodes of $C_{I}'$ then $C_{I}'$ is satisfied, since these nodes of $C_{I}'$ compute the same values as those that  $V_1, V_2, \dots, V_N$ compute in $C_{I}^{o}$. 
	
	As it was shown in Lemma \ref{lem:C'=C}, the output of $C_{I}'$ computes the same output as $C_{I}$, which computes the function $G_{I}$. Thus, for $x^*$ it holds that $G_{I}(x^*) = x^*$, i.e. it is a fixed point of $G_{I}$. Recall now that the fixed points of $G_{I}$ are precisely the Nash equilibria of instance $I$ of the initial ``$d$-player Nash equilibrium'' problem. Since, due to \cite{EY10}, ``$d$-player Nash equilibrium'' is \FIXP-complete, it follows that $(n,n)$-\CH is \FIXP-hard.    
\end{proof}

\section{Proof of Lemma \ref{lem:Feas_is_etr-c}}\label{app:Feas_is_etr-c}

Let us define the 
constrained version of \etr, denoted $\betr$, where the polynomials
are over $[0,1]^n$. It is easy to see that $\betr \subseteq \etr$; an arbitrary \betr instance $\exists (X_1, \dots, X_m) \in [0,1]^m \cdot \Phi$, where $\Phi$ is the \betr formula, can be written as the following \etr instance $\exists (X_1, \dots, X_m) \in \reals^m \cdot \Phi \bigwedge_{i=1}^{m} \left( \left(X_{i} \geq 0 \right) \wedge \left(X_{i} \leq 1 \right) \right)$.

Lemma 3.9 of \cite{scha13} proves that the problem of deciding whether a family of polynomials $p_i : \reals^n \to \reals$, $i \in [k]$ has a common root in the $n$-dimensional unit-ball (with center $0^n$ and radius 1) is \etr-complete. Furthermore, this holds even when all $p_i$'s have maximum sum of variable exponents in each monomial equal to 2. Since the $n$-dimensional unit-ball is inscribed in the $n$-dimensional unit-cube, the aforementioned result implies that \bconj is \etr-hard, therefore $\etr \subseteq \betr$. Consequently, we get the following.
\begin{theorem}\label{thm:betr-etr}
	$\betr = \etr$.
\end{theorem}

Now we will prove that \bfeas is \etr-hard by reducing \bconj to it. We do this by a standard way of turning a conjunction of polynomials into a single polynomial, so that all zeros of the conjunction are exactly the same as the zeros of the single polynomial. Consider an instance of \bconj. Let us define the function $q = (p_1)^2 + (p_2)^2 + \dots + (p_k)^2$, which has again domain $[0,1]^n$. The instance of \bfeas with function $q$ has exactly the same solutions as these of the \bconj instance. Therefore \bfeas is \etr-complete, even when $q$ has maximum sum of variable exponents in each monomial equal to 4.

%first we present a polynomial-time reduction from the \etr-complete problem \feas to an intermediate problem \bconj, then we reduce \bconj to another intermediate problem \bfeas, and finally we reduce the latter to $(n,k)$-\CH. A straightforward corollary is that \bconj and \bfeas are \etr-complete, a result that we believe is of independent interest. \themis{I will write the reduction from \bconj to \bfeas.}

%%%%%%%%%%%%%%%%%%%%%%%%%%%%%%%%%%%%%%%%%%%%%

\section{Proof of Theorem \ref{thm:CH_is_etr-c}}\label{app:CH_is_etr-c}

As we show in Theorem \ref{thm:CHinETR}, $(n,k)$-\CH is in \etr. In this section we prove that $(n,n-1)$-\CH is \etr-hard, implying that it is complete for \etr. This complements the results of \cite{FFGZ18}, where it was established that $(n,n-1)$-\CH is \NP-hard even when a solution is required to be $1/poly(n)$-approximately correct, i.e.
it allows $|F_i(A_+)- F_i(A_-)| \leq \eps$ for every agent $i$, where $\eps = 1/poly(n)$.

We present a polynomial-time reduction from the \etr-complete problem \bfeas to $(n,n-1)$-\CH.
Suppose we are asked to decide an arbitrary instance of $\textsc{Feasible}_{[0,1]}$, i.e. the existential sentence 
\begin{align}\label{Feasible1}
	\left( \exists X \in [0,1]^N \right) (p(X) = 0),
\end{align}
where $X := (X_1, \dots, X_N) \in [0,1]^N$ and $p$, is a polynomial function of $X_1, \dots, X_N$ written in the standard form (a sum of monomials with integer coefficients).
Consider the integer coefficients $C_1, \dots, C_l$ of $p$, where the number of terms of the polynomial is $l$. We consider all of the coefficients to be positive, where some of them may be preceded by a ``$+$'' or a ``$-$''. Also, let us normalize the coefficients and create new ones $c_1,\dots, c_l$, where
\begin{align*}
	c_j := \frac{C_j}{l \cdot C_{max}}, \quad j \in [l],
\end{align*}
where $C_{max}:= \max_{j} C_j$.
Note that our new polynomial $q(X)$ which uses the new coefficients has exactly the same roots as $p(X)$. Also, note that $c_j \in (0,\frac{1}{l}]$ for every $j \in [l]$, a fact that will play an important role at the last steps of our reduction.

Now, let us split polynomial $q$ into two polynomials $q_1$ and $q_2$, such that 
\begin{align*}
	q(X):= q_{1}(X) - q_{2}(X),
\end{align*}
and both $q_1$ and $q_2$ are sums of \textit{positive} terms; $l_1$ and $l_2$ terms of $q_1$ and $q_2$ respectively, where $l = l_1 + l_2$. In particular, 
\begin{align*}
	q_{1}(X) &:= \sum_{j=1}^{l_1} r_j(X), \\
	q_{2}(X) &:= \sum_{j=l_1 + 1}^{l_2} r_j(X),
\end{align*}
where $r_j(X):= c_j \cdot X_{1}^{d_{1j}} \cdot \dots \cdot X_{N}^{d_{Nj}}$ is the term $j \in [l]$ and $d_{ij}$ is the exponent of variable $X_i$, $i \in [N]$, in the $j$-th term.
Eventually, the existential sentence, equivalent to (\ref{Feasible1}), that we ask to decide is
\begin{align*}
	\left( \exists X \in [0,1]^N \right) (q_{1}(X) = q_{2}(X)).
\end{align*}

Let us construct the algebraic circuit that takes as input the tuple $X$ and computes the value of $q_{1}(X)$. This circuit needs only to use gates in $\{G_{\zeta}, G_{+}, G_{*\zeta}, G_{*}, G_{()^2}\}$, where $\zeta \in \mathbb{Q} \cap (0,1]$. To see why, observe that since every $X_{i} \in [0,1]$, $i \in [N]$, any multiplication between them by a $G_{*}$ gate is done properly (the gate's inputs/output are in $[0,1]$), and obviously the same holds for $G_{()^2}$. Also, note that due to our downscaled coefficients $c_j$, it is $c_j \leq 1/2$ for every $j$, and also
\begin{align}\label{ub_of_sum}
	\sum_{j=1}^{l_1} r_j(X) \leq l_1/l \leq 1.
\end{align}   
Therefore, we guarantee that any of the $l_1 - 1$ additions of the terms $r_j$ of $q_1$ by a $G_{+}$ gate is done properly, (inputs in $[0,1/2]$ and output in $[0,1]$).  Similarly, we construct a circuit that computes $q_2$.

At this point we are ready to prove Theorem \ref{thm:CH_is_etr-c}.

\begin{proof}
	
	Let us construct a $(n,n-1)$-\CH instance, where $n$ is to be defined later. 
	In Section \ref{subsec: cir-to-ch} we have shown how to construct an equivalent circuit to the one that computes $q_1, q_2$, called ``special circuit'', that
	\begin{itemize}
		\item uses only gates $G_{\zeta}, G_{+}, G_{*\zeta}, G_{()^2}, G_{-}^{[0,1]}, G_{*2}^{[0,1]}$,
		\item every $G_{\zeta}$ and $G_{* \zeta}$ has $\zeta \in \mathbb{Q} \cap (0, 1]$,
		\item for every input $x \in [0, 1]^N$, all intermediate values computed by the
		circuit lie in $[0, 1]$.
	\end{itemize}
	For the constraints of the above types of gates, see Tables \ref{tbl:gates}, \ref{tbl:special_gates}.
	
	Let the number of gates in that special circuit be $r := poly(N)$. Consider the last two nodes of the special circuit whose outgoing edges are $q_1$ and $q_2$ respectively. Without loss of generality, we name them $v_{r-1}$ and $v_{r}$ (see Figure \ref{pl-n}).
	
	\begin{figure}
		\centering
		\includegraphics[width=0.35\linewidth]{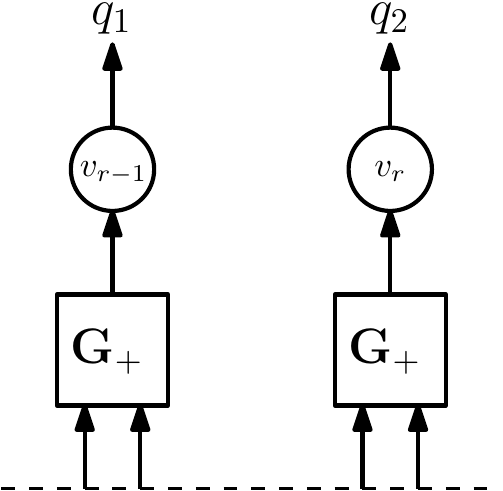}
		\caption{The last two nodes of the special circuit.}\label{pl-n}
	\end{figure}

	By Lemma \ref{lem:circuit-embed} and the construction described in its proof (Section \ref{app:circuit-embed}), we embed the special circuit in a \CH instance. This instance now consists of $4r$ agents, since to each node $i \in [r]$ correspond 4 agents: $ad_{i}, mid_{i}, cen_{i}$ and $ex_i$ with valuation functions described by \eqref{val-functions}.
	
	According to the embedding described in Section \ref{app:circuit-embed}, a tuple $(z_{1}^*, \dots, z_{r}^*)$ of values that satisfies the special circuit, corresponds to a $(4r,4r)$-\CH solution, i.e. a tuple $(t_{1}^*, \dots, t_{4r}^*)$ with $0 \leq t_{1}^* \leq \dots \leq t_{4r}^* \leq 12r$, of the \CH instance we constructed, and vice versa. As shown in detail in Section \ref{app:circuit-embed}, every value $z_{i}^*$ in a solution can be translated to 4 cuts $t_{4i-3}^*, t_{4i-2}^*, t_{4i-1}^*, t_{4i}^*$ in the \CH solution by the transformation in Section \ref{par:cir-to-ch}. Conversely, a 4-tuple $(t_{4i-3}^*, t_{4i-2}^*, t_{4i-1}^*, t_{4i}^*)$ of cuts in a \CH solution can be translated to a single value $z_{i}^*$ by the simple transformation $z_{i}^* = t_{4i}^* - v_{i,l}^{+}$ in Section \ref{par:ch-to-cir}.
	
	Let us now introduce a $(4r + 1)$-st additional agent, named \textit{finis} (from the Latin word for ``end'') who does not correspond to any node. The valuation function of this agent is non-zero only in the intervals $v_{r-1}^{+}$ and $v_{r}^{-}$ and, in particular is the following,
	\begin{align}\label{finis}
		finis(t) &= \begin{cases}
			1 , \quad &t \in v_{r-1}^{+} \cup v_{r}^{-} \\
			0 , \quad &\text{otherwise}.
		\end{cases}
	\end{align}
	Eventually, the number of agents in the embedding is $n := 4r + 1$.
	
	We will show that the answer to the arbitrary \bfeas instance \eqref{Feasible1} is ``yes'', if and only if the answer to the $(n,n-1)-\CH$ problem is ``yes'', i.e. there exists a $(n-1)$-cut that satisfies $n$ agents. 
	
	Suppose that there exists a solution $X^*:=(X_{1}^*, \dots, X_{N}^*) \in [0,1]^{N}$ of \eqref{Feasible1}, which equivalently means that $q_{1}(X^*) = q_{2}(X^*)$. Then, by the correct construction of our special circuit (following the procedure in Section \ref{subsec: cir-to-ch}) which uses $r$ gates and computes $q_1$ and $q_2$, there is a tuple $z^*:=(z_{1}^*, \dots, z_{r}^*)$ that satisfies it. Let, without loss of generality, $(z_{1}^*, \dots, z_{N}^*) := (X_{1}^*, \dots, X_{N}^*)$. Then it holds that $q_1(z_{1}^*, \dots, z_{N}^*) = q_2(z_{1}^*, \dots, z_{N}^*)$, therefore $z_{r-1}^* = z_{r}^*$. 
	
	According to the aforementioned translation to cuts, in the \CH instance there will be a cut $t_{4(r-1)}^* = v_{r-1,l}^{+} + z_{r-1}^*$ in interval $v_{r-1}^{+}$ (i.e. a positive cut), and another one in $t_{4r-1}^* = v_{r,l}^{-} + z_{r}^*$ in interval $v_{r}^{-}$ (i.e. a negative cut). From the valuation function \eqref{finis} of agent $finis$, we can see that her positive total valuation equals her negative total valuation, since $z_{r-1}^* \cdot 1 + (1 - z_{r}^*) \cdot 1 = (1 - z_{r-1}^*) \cdot 1 + z_{r}^* \cdot 1$ holds from $z_{r-1}^* = z_{r}^*$.
	Therefore $finis$ is satisfied. Also, the agents $ad_i, mid_i, cen_i, ex_i$ for all $i \in [r]$ are satisfied as argued in Section \ref{app:circuit-embed}, and the answer to $(n,n-1)-\CH$ is ``yes'', since we have $4r+1$ agents satisfied by $4r$ cuts.
	
	Suppose now that there exists a $4r$-cut $(t_{1}^*, \dots, t_{4r}^*)$ with $0 \leq t_{1}^* \leq \dots \leq t_{4r}^* \leq 12r$ that is a solution of the $(n,n-1)$-\CH instance we constructed, where $n := 4r + 1$. As argued in Section \ref{app:circuit-embed}, if the $ad_i, mid_i, cen_i, ex_i$ agents for $i \in [r]$ are satisfied then each of $cen_i, ex_i$ agents imposes a cut in interval $v_{i}^{-}$ and $v_{i}^{+}$ respectively. The cuts in intervals $v_{i}^{+}$, for all $i \in [r]$ can be translated back to values $z_i^*$, which successfully compute the values of the circuit, i.e. they satisfy the circuit. There are also two interesting cuts $t_{4(r-1)}^*$ and $t_{4r-1}^*$ imposed by $ex_{r-1}$ and $cen_{r}$ respectively which satisfy agent $finis$. Since this agent is satisfied with no additional cut, it holds that $z_{r-1}^* \cdot 1 + (1 - z_{r}^*) \cdot 1 = (1 - z_{r-1}^*) \cdot 1 + z_{r}^* \cdot 1$, or equivalently $z_{r-1}^* = z_{r}^*$. Since $z_{r-1}^*$ and $z_{r}^*$ correspond to the value of the circuit at $q_1$ and $q_2$ respectively, for the circuit's inputs $(z_{1}^*, \dots, z_{N}^*)$ it holds that $q_{1}(z_{1}^*, \dots, z_{N}^*) = q_{2}(z_{1}^*, \dots, z_{N}^*)$. Equivalently, $q(z_{1}^*, \dots, z_{N}^*) = 0$, and equivalently $p(z_{1}^*, \dots, z_{N}^*) = 0$. Therefore, we have found values that satisfy \eqref{Feasible1}, and the answer to \bfeas is ``yes''.     
\end{proof}

\section{Conclusion and Open Problems}

In this work we studied the complexity of exact computation of a solution to
\CH. We introduced the class \BU which captures all problems that are
polynomial-time reducible to the Borsuk-Ulam problem. We showed that the
complexity of $(n,n)$-\CH is lower bounded by \FIXP and upper bounded by \BU. A
tight result on the complexity of $(n,n)$-\CH is the major open problem that
remains. We believe that the problem is \BU-complete. Such a result would
establish \BU as a complexity class that has a complete natural problem. We also
believe that the best candidates of \BU-complete problems are function problems
whose solution existence is provable by the Borsuk-Ulam theorem, but not known
to be 
provable by any weaker one, for example, Brouwer's Fixed Point theorem. One such
is the Ham Sandwich problem \cite{S1942} whose complexity is still unresolved:
given $n$ compact sets in $\mathbb{R}^n$, find an $(n-1)$-dimensional hyperplane
that bisects all of them. 

Our result that \linBU = \PPA is analogous to the result of
\cite{EY10} which shows that \linFIXP = \PPAD. In \cite{M14}, the classes
\texttt{kD}-\linFIXP, $k \geq 1$ are implicitly defined as the subclasses of
\FIXP which contain all problems that can be described by \linFIXP circuits with
$k$ inputs. It was shown in \cite{M14} that \texttt{2D}-\linFIXP = \PPAD, which
uncovered a sharp dichotomy on the complexity of \linFIXP problems;
$\texttt{1D}-\linFIXP \subseteq \p$ (by \cite{AGMS11}) while
\texttt{kD}-\linFIXP = \PPAD for $k \geq 2$. An interesting open problem is to
consider the analogue of these classes in \linBU, namely \texttt{kD}-\linBU, and
study the complexity of the problem depending on values of $k$.

\bibliographystyle{plain}
\bibliography{references}

\end{document}